\documentclass[a4paper]{article}

\usepackage{wrapfig}
\usepackage[utf8]{inputenc}
\usepackage{graphicx}
\usepackage{amsmath}
\usepackage{amsthm}
\usepackage{verbatim}
\usepackage{comment}
\usepackage{amssymb}
\usepackage{todonotes}
\usepackage{enumerate}
\usepackage[noend, ruled]{algorithm2e}
\usepackage{authblk}
\usepackage{microtype}
\DontPrintSemicolon

\newtheorem{theorem}{Theorem}
\newtheorem{lemma}[theorem]{Lemma}

\newtheorem{observation}[theorem]{Observation}
\theoremstyle{remark}

\newtheorem{example}[theorem]{Example}
\newtheorem{corollary}[theorem]{Corollary}
\newtheorem{claim}[theorem]{Claim}

\newcommand{\Oh}{\mathcal{O}}

\newcommand{\Seq}{\mbox{\sf{SEQ}}}
\newcommand{\PATH}{\mbox{\sf{PATH}}}
\newcommand{\SeedGraph}{\mbox{\sf{SeedGraph}}}
\newcommand{\ST}{\mbox{\sf{ST}}}
\newcommand{\SW}{\mbox{\sf{SW}}}
\newcommand{\W}{\mathbf{W}}
\newcommand{\V}{\mathbf{V}}

\newcommand{\R}{\mathcal{R}}
\newcommand{\G}{\mathcal{G}}

\usepackage[hidelinks]{hyperref}
\usepackage[capitalise]{cleveref}
\begin{document}
\title{Syntactic View of Sigma-Tau Generation \texorpdfstring{\\}{}
of Permutations}

%On Sigma-Tau Generation of Permutations}
\author{Wojciech Rytter}
\author{Wiktor Zuba}
  \affil{\normalsize Faculty~of Mathematics, Informatics and Mechanics,
    University of Warsaw, Warsaw, Poland\\
    \texttt{[rytter,w.zuba]@mimuw.edu.pl}}

\date{\vspace{-0.5cm}}

\maketitle              % typeset the header of the contribution

\begin{abstract}
We give a syntactic view of the Sawada-Williams
$(\sigma,\tau)$-generation of permutations.
The corresponding sequence of $\sigma\-\tau$-operations,
of length $n!-1$ is shown to be highly
compressible: it has $\Oh(n^2\log n)$ bit description.
Using this compact description we design
fast algorithms for ranking and unranking permutations.
\end{abstract}

\section{Introduction}
 We consider permutations of the set $\{1,2,...,n\}$, called here $n$-permutations.\\
For an $n$-permutation $\pi=(a_1,...,a_n)$ denote:

 $\sigma(\pi)=(a_2,a_3,...,a_n,a_1),\;
\tau(\pi)=(a_2,a_1,a_3,...,a_n).$ 

\smallskip\noindent In their classical book on combinatorial algorithms 
Nijenhuis and Wilf asked in 1975 
if all $n$-permutations can be generated, each exactly once, using 
in each iteration a single operation $\sigma$ or $\tau$.
This difficult problem was open for more than 40 years. 
Very recently Sawada and Williams 
presented an algorithmic solution at the conference SODA'2018.
In this paper we give new insights into their algorithm
by looking at the generation from syntactic point of view.

Usually in a generation of combinatorial objects of size $n$ 
we have a starting object and
some  set $\Sigma$ of very local operations. Next object results by
applying an operation from $\Sigma$, 
the generation is efficient iff each local
operation uses small memory and time. Usually the sequence 
of generated objects is
exponential w.r.t. $n$.
From a syntactic point of view the
generation globally can be seen as a very large word in
the alphabet $\Sigma$ describing the sequence
of operations. It is called the {\it syntactic sequence}
of the generation. Its textual properties can help to understand better the generation and
to design efficient ranking and unranking. Such syntactic approach was 
used for example by Ruskey and Williams in generation of  (n-1)-permutations of an n-set
 in \cite{DBLP:journals/talg/RuskeyW10}.

\noindent Here we are interested whether the syntactic sequence
is highly compressible.
We consider compression in terms of Straight-Line Programs 
({\it SLP}, in short), which represent large words by recurrences,
see \cite{DBLP:conf/icalp/Rytter04}, using operations of concatenation.
We construct SLP with $\Oh(n^2)$ recurrences, which has $\Oh(n^2\log n)$ bit description.

The syntactic sequence for some generations is highly compressible and
for others is not. For example in case of reflected binary Gray code 
of rank $n$ each local operation is the position of the changed bit.
Here $\Sigma=\{1,2,...,n\}$ and the syntactic sequence $T(n)$ is described by 
the short SLP of only $\Oh(n)$ size:
 $T_1=1;\ \ T(k)\,=\, T(k-1),\,k,\,T(k-1) \ \mbox{for}\ 2\le k \le n.$

In case of de Bruijn words of length $n$ each operation 
corresponds to a single letter appended at the end.
However in this case the syntactic sequence  is not highly compressible
though the sequence can be iteratively computed in a very simple way,
see \cite{DBLP:journals/dm/SawadaWW16}.
In this paper we consider 
the syntactic  sequence $\Seq_n$ (over alphabet $\Sigma=\{\sigma,\,\tau\}$)
 of Sawada-Williams
$\sigma\tau$-generation of permutations presented in \cite{DBLP:conf/soda/SawadaW18,sawadahomepage}.
An SLP of size $\Oh(n^2)$ describing $\Seq_n$ is given in this paper.
The $\sigma\tau$-generation of $n$-permutations
by Sawada and Williams can be seen as a Hamiltonian path $\SW(n)$
in the Cayley graph $\G_n$.
The   nodes of this graph are permutations and the
edges correspond to operations $\sigma$ and $\tau$.

\noindent We assume that (simple) arithmetic operations used in the paper 
are computable in constant time.
\paragraph{\bf Our results.} We show:
\begin{enumerate}
\item 
$\Seq_n$ 
can be represented 
by the straight-line program of $\Oh(n^2)$ size:

\smallskip
\begin{itemize} 
\item $
\W_0=\sigma,\ \ \
\W_k \;=\; \tau\ \cdot\ \prod_{i=1}^{n-2}\,\sigma^i\,
\W_{\Delta(k,i)}\,\gamma_{n-2-i}
$

\smallskip\hspace*{0.5cm}
$\ \mbox{for}\ 1\le k < n-3;$

\medskip
\item
$\V_n \;=\; \gamma_{n-3}\ \cdot\ \prod_{i=2}^{n-3}\,\sigma^i\,
\W_{\Delta(n-3,i)}\,\gamma_{n-2-i} \ \cdot \ \sigma^{n-1};
$

\medskip
\item
$ \Seq_n\;=\; \gamma_1^{n-2}\sigma^2\;(\V_n\,\tau)^{n-2}\; \V_n.$
\end{itemize}

\smallskip where $\Delta(k,i)=\min(k-1,n-2-i)$ and $\gamma_k\,=\, \sigma^k\tau$.

\medskip
\item  {\bf Ranking:} using compact description of $\Seq_n$ 
 the number of  steps (the rank of the permutation) 
needed to obtain a given permutation from
a starting one can be computed in time $\Oh(n\sqrt{\log n})$
using inversion-vectors of permutations.

\smallskip
\item
{\bf Unranking:} again using $\Seq_n$
the $t$-th permutation generated by 
$\Seq_n$ can be computed in $\Oh(n\frac{\log n}{\log\log n})$ time.
\end{enumerate}

\section{Preliminaries}
Denote by $cycle(\pi)$ all permutations 
cyclically equivalent to $\pi$.
Sawada and Williams introduced an ingenious concept of 
a seed: a {\it shortened permutation} 
representing a group of $(n-1)$ cycles.
Informally it represents a set of permutations which are cyclically
equivalent {\it modulo} one fixed element, which can appear in any place.

\medskip\noindent
Let $\oplus$ denote a modified addition modulo $n-1$, where $n-1\oplus 1 = 1$.
It gives a cyclic order of elements $\{1,...,n-1\}$. We write $a\ominus 1\,=b$
iff $b\oplus 1=a$.

\noindent Formally a  {\it seed} is a $(n-1)$ tuple of
distinct elements of $\{1,2,...,n\}$ of the form 
$\psi=(a_1,a_2,...,a_{n-1})$, such that $a_1=n$ and 
$(a_1,a_2\oplus1,a_2,...,a_{n-1})$ is a permutation. 
The element $x=mis(\psi)=a_2\oplus1$ is called a {\it missing} element.

Denote by $perms(\psi)$ 
the set of all $n$-permutations resulting by making 
a single insertion of $x$ into any position in  $\psi$, and making cyclic shifts.
The sets $perms(\psi)$ are called {\it packages}, the seed $\psi$ is
the {\it identifier} of its package $perms(\psi)$.
One of the main tricks in the Sawada-Williams construction is 
the requirement that the missing element equals $a_2\oplus1$.
In particular this implies the following:
\begin{observation} A given $n$-permutation belongs to one or two
packages. We can find identifiers of these packages
in linear time.
\end{observation}

\medskip The algorithm  of Sawada and Williams starts 
with a construction of a large and a small cycle
(covering together the whole graph). The graph consisting
of these two cycle is denote here by $\R_n$.
The small cycle is very simple. 
Once $\R_n$ is constructed the Hamiltonian path is very easy: 
In each cycle one $\tau$-edge is removed (the cycles become
simple paths),  then the cycles are connected by
adding one edge to $\R_n$. 

\subsection{Structure of seed graphs}\label{subs 2.1}
First we introduce seed-graphs.
Define  the {\it seed-graph} of the seed $\psi$, denoted here by 
$\SeedGraph(\psi)$ (denoted by $Ham(\psi)$ in \cite{DBLP:conf/soda/SawadaW18}), 
as the graph consisting of edges {\it implied}
 by the seed $\psi$. The set of nodes consists of $perms(\psi)$, the
set of edges consists of almost all $\sigma$-edges
between these nodes (except the edges of the form $(*,x,*,...,*)\rightarrow (x,*,...,*,*)$), but the set of $\tau$-edges consists only
of the edges of the form $(*,x,*,...,*) \rightarrow (x,*,*,...,*)$,
where $x$ is the {\it missing} element.
see Figure~\ref{SeedGraphpsi)}.

\bigskip
\begin{figure}[h]
\centering
\vspace*{-.7cm}
\centerline{\includegraphics[width=5.cm]{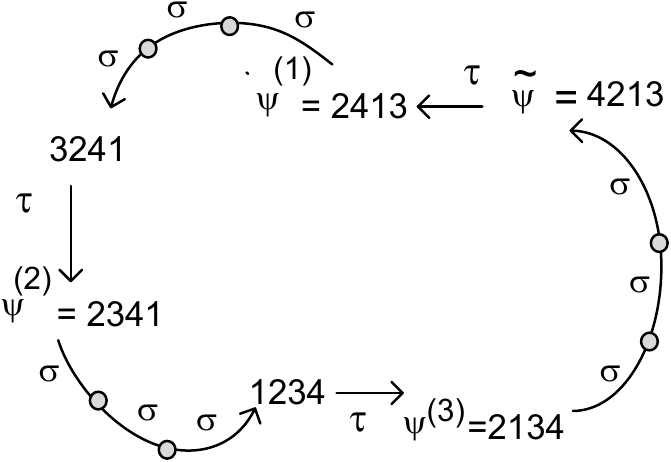}}
\caption{Structure of $SeedGraph(\psi)$, where $\psi=(4,1,3)$, 
$mis(\psi)=2$.}\label{SeedGraphpsi)}
\end{figure}
\vspace*{-0.5cm}

\noindent For a seed $\psi=(a_1,a_2,...,a_{n-1})$ with $mis(\psi)=x$,
let

\smallskip
\centerline{ $\psi^{(n-1)}=(x,a_2,...,a_{n-1},a_1),\; \widetilde{\psi}\,=\, (a_1,x,a_2,a_3,...,a_{n-1})$.}

\smallskip\noindent
For $1\leq i\leq n-1$ denote $\psi^{(i)\;}=\;
\gamma_{n-1}^i(\psi^{(n-1)}).$ 
In other words $\psi^{(i)}$, for $n>i>0$,  is the word $\psi$ right-shifted
by $i-1$ and with $x$ added at the beginning.
Observe that:
$\gamma_{n-1}(\psi^{(i)})=\psi^{(i+1)}\ \mbox{for}\ 0<i<n-1.$
\begin{example} For  $\psi=(5,3,2,1)$ we have 
$\widetilde{\psi}=(5,4,3,2,1),\; \psi^{(1)}=(4,5,3,2,1),$\\
$\psi^{(2)}=(4,1,5,3,2),
\; \psi^{(3)}=(4,2,1,5,3),\; \psi^{(4)}=(4,3,2,1,5)
.$
\end{example}

\noindent Each $perms(\psi)$ can be sequenced easily as a simple cycle in 
$\G_n$. 
Two seeds $\phi,\psi$ are called {\it neighbors} iff $perms(\phi)\cap perms(\psi)\ne \emptyset$.
The permutations of type $\psi^{(i)}$ play crucial role
as {\it connecting} points between packages of
neighboring seeds.

\begin{observation} Two distinct seeds $\phi,\psi$ are neighbors iff
$mis(\phi)=mis(\psi)\oplus 1\ \mbox{or}\ mis(\psi)=mis(\phi)\oplus 1$,
and after removing both $mis(\psi),\,mis(\phi)$ from $\phi$ and $\psi$
the sequences $\phi,\psi$ become identical.\\
\end{observation}

\subsection{The pseudo-tree \texorpdfstring{$\ST_n$}{STn} of seeds}\label{subs 2.2}
For a seed $\psi=a_1a_2...a_{n-1}$ denote by 
$height(\psi)$ the maximal length $k$ of a prefix of $a_2,a_3,...,a_{n-1}$
such that $a_i=a_{i+1}\oplus 1$ for $i=2,3,...,k$.
For example $height(94326781)=3$ (here the {\it missing} number is 5).
For each two neighbors we distinguish one of them as a parent of the second one
and obtain a tree-like structure called a {\it pseudo-tree} denoted by $\ST_n$.
If $height(\psi)>1$ and $mis(\psi)=mis(\beta)\oplus1$
we write $parent(\beta)=\psi$. 
Additionally if $\sigma^i(\psi^{(i)})=\widetilde{\beta}$ 
we write $son(\psi,i)=\beta$ and we say that $\beta$ is the $i$-th son of $\psi$. 

\smallskip
\noindent The function $parent$ gives the tree-like graph of the set
of seeds, it is a cycle with hanging subtrees rooted at nodes
of this cycle. The set of seeds on this cycle is denoted by $Hub_n$. 
For example 

\centerline{$Hub_6=\{(6,5,4,3,2),(6,4,3,2,1),(6,3,2,1,5),(6,2,1,5,4),(6,1,5,4,3)\}.$}

\noindent Due to Lemma~\ref{same Hamiltonian} we have:
\begin{observation}
If $\psi \notin Hub_n$ then \underline{all} $\tau$-edges of $\SeedGraph(\psi)$ are in 
$PATH(n)$.
\end{observation}

\medskip\noindent
For $\psi\notin Hub_n$ 
let $Tree(\psi)$ be the subtree of $\ST_n$ rooted at $\psi$
including $\psi$ and nodes from which $\psi$ 
is reachable by $parent$-links.

\subsection{A version of Sawada-Williams algorithm}
%\smallskip\noindent 
\noindent We say that an edge $u\rightarrow v$ conflicts with $u'\rightarrow v'$
iff $u=u', v\ne v'$.
{Non-disjoint packages $\phi,\psi$
can be joined into a simple cycle
 by removing two $\sigma$-edges conflicting with $\tau$-edges.
}

\smallskip
\centerline{\includegraphics[width=4.5cm]{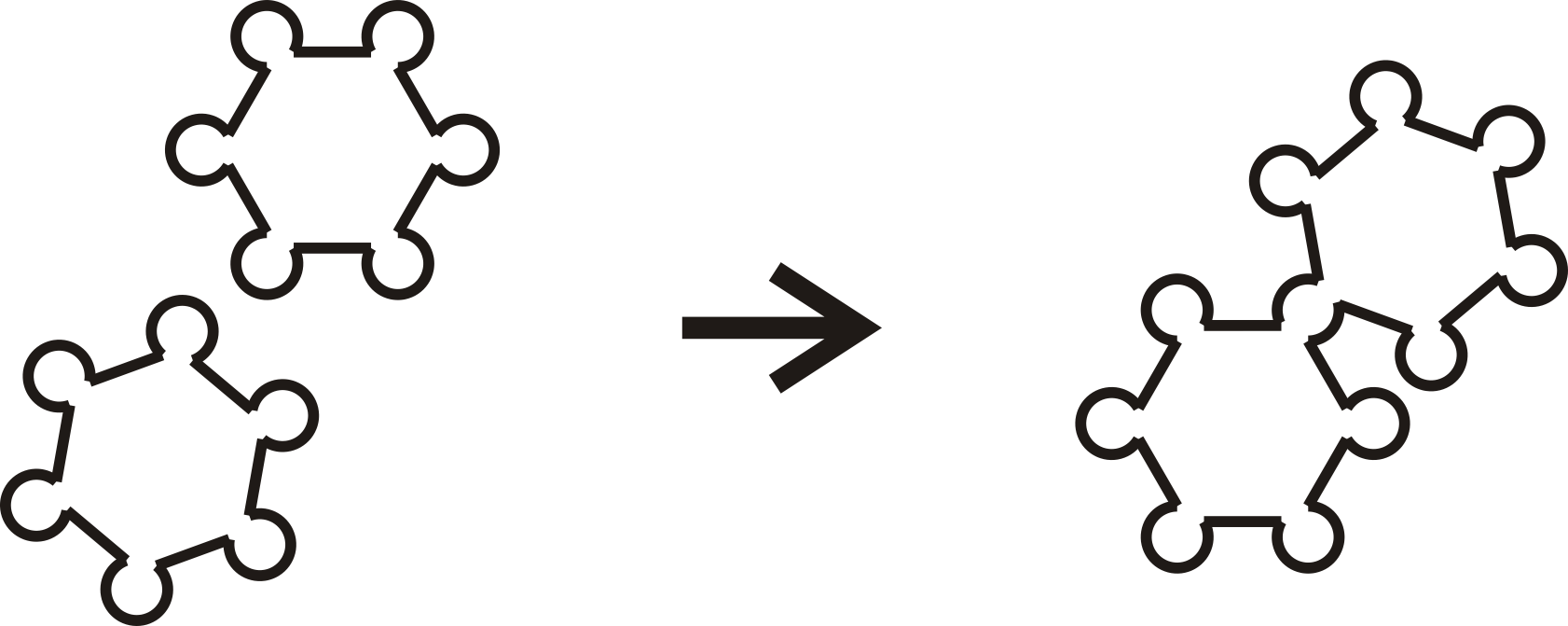}}

\smallskip\noindent
 By a union of graphs we mean set-theoretic union of nodes and
set-theoretic union of all edges in these graphs.
\\
Denote by $\R_n$ the graph $ \bigcup_{\psi}\,\SeedGraph(\psi)$ in which
we removed all $\sigma$-edges conflicting with $\tau$-edges.
The $\tau$-edges have priority here.
\noindent A version of the construction of a Hamiltonian path by 
Sawada-Williams, denoted by $\SW(n)$, can be written informally as:

\medskip\noindent
\begin{center}
\begin{minipage}{10.3cm}
{\bf Algorithm } Compute $\PATH(n)$;

\medskip
$P := \bigcup_{\psi\in SEEDS(n)}\,\SeedGraph(\psi)$

\smallskip
remove from $P$ all $\sigma$-edges conflicting with $\tau$-edges
in $P$

\smallskip\hspace{0.3cm}
$\pi:=(n,n-1,...,1)$;
add to $P$ the edge  $\pi\rightarrow \sigma(\pi)$

\smallskip\hspace{0.3cm}
remove edges $\pi \rightarrow \tau(\pi),\ 
\tau(\sigma(\pi))\rightarrow \sigma(\pi)$

\smallskip
{\bf return} $P$\ 
\{$P$ is now  a Hamiltonian path $\tau(\pi)\rightarrow^* \tau(\sigma(\pi))$\,\}
\end{minipage}
\end{center}
\begin{lemma}\label{same Hamiltonian}
$\PATH(n)=\SW(n)$.
\end{lemma}
\begin{proof}
 To prove that the paths are the same it is enough to prove that  both begin in the same place and that $\tau$ edges are used from exactly the same vertices.
 The particular Hamiltonian path in \cite{DBLP:conf/soda/SawadaW18} is described in terms of a function $next$, which for any vertex assigns a next one on the path (by returning $\sigma$ or $\tau$).
 The function $next$ for a given permutation $\pi=(p_1,p_2,...,p_n)$ produces a $\tau$-edge when $p_2=r\oplus 1$ (unless $\pi=(n,n-1,...,1)$ or $p_2=n$), where
 $r$ is cyclically first element after $n$ jumping over $p_2$ (equal to $p_3$ if $n=p_1$).
 
 %The definitions of $Hub_n$ and $\psi^{(i)}$ are given in subsections \ref{subs 2.1} and \ref{subs 2.2}.
 
\smallskip
The condition of being equal to $r\oplus 1$ in permutation $\pi$ is exactly the same condition as being the missing element of a seed $\psi$ such that $\pi\in perms(\psi)$
 (for one of two seeds if $r=p_3=p_2\ominus 1$).
 For a permutation $\pi\in perms(\psi)$ such that $height(\psi)<n-3$ the missing element of $\psi$ happens to be the first element of the permutation if and only if
 $\pi\in\bigcup\limits_{i=1}^{n-1}\{\psi^{(i)}\}$.
 $\bigcup\limits_{i=1}^{n-1}\{\psi^{(i)}\}$ is a set of those permutations from $perms(\psi)\cap bunch(\psi)$, whose ingoing edge in $PATH(n)$ is a $\tau$-edge.
 As $\tau$-edges exchange the first two elements hence both approaches
 describe the same sets of edges (unless $\psi$ belongs to the $Hub_n$).
 If the permutation belongs to the package $perms(\psi)$ for a hub seed $\psi$, the only difference from the previous case is that it can belong to the first part of the path of length $2n-2$
 (including the special permutation $(n, n-1, ..., 1)$).
 
 The first $2n-2$ permutations (the ones that are cyclically equivalent to $(n-1,n-2,...,1)$ after removing element $n$ which appears on first or second position)
 generate the same pattern in both constructions (alternation of $\sigma$ and $\tau$ edges).
 All other permutations from those packages follow the same rules as the ones from packages corresponding to seeds with height $\le n-4$, with the difference that permutations $\psi^{(i)}$ are not explicitly named
 (and that $\psi^{(1)}$ play a different role being the ones from the beginning of the path).
 In this way we have shown that $SW(n)=PATH(n)$.
\end{proof}
%----------------------------------------
%--------------------------------------------------
\section{Compact representation of bunches of permutations}
\noindent Our aim is to give a syntactic version of $\PATH(n)$: the sequence
$\Seq_n$ of $\sigma\tau$-labels of $\PATH(n)$ represented compactly.
We have to investigate more carefully the structure of seed-graphs and their interconnections.
We introduce the basic components of $PATH(n)$: groups of permutations 
corresponding to a subtree of seeds which are not on the cycle in the 
pseudo-tree. For $\psi \notin Hub_n$ define $$bunch(\psi)\;=\ 
\bigcup_{\beta\in Tree(\psi)}\,
perms(\beta)\ -\ cycle(\widetilde{\psi})\ \cup\ \{\widetilde{\psi},\psi^{(n-1)}\}.$$

In other words $cycle(\widetilde{\psi})$ connects $bunch(\psi)$ with the 
{\it "outside world"}, only through $\widetilde{\psi},\psi^{(n-1)}$.

\noindent We start with properties of local interconnection between
two packages.

\begin{lemma}\label{seed intersections}\mbox{ \ } Two seeds $\phi\ne \psi$ are neighbors iff 
one of them  is the parent of another one.
If $\phi=parent(\psi)$
then $perms(\phi)\cap perms(\psi)$ is 
the $\sigma$-cycle containing both $\widetilde{\psi}$ and 
$\phi^{(i)}$, for some $i$, 
and has a structure as shown in Figure~\ref{Sept26}(A), where $\psi$ is the $i$-th son of $\phi$.
If $height(\phi)=k<n-3$ then $height(\psi)=\Delta(k,i)$.
Furthermore $son(\phi,i)$ exists for all $i\in\{1,...,n-3\}$.
\end{lemma}

\begin{proof}
 Assume that $\pi=(p_1,p_2,...,p_n)\in perms(\phi)$. Without loss of generality we can assume that $p_1=n$ (as cyclically equivalent permutations belong to the same packages).
 After removing any element from $\pi$ the first element after $n$ is either $p_2$ or $p_3$ in both cases we obtain a valid seed if and only if the removed element is
 greater by one than that element following $n$. When removing element $p_2\oplus 1$ we always receive a seed which package contains $\pi$ and if we remove $p_2$ this is only the case if $p_2=p_3\oplus 1$.
 It means that if $\pi\in perms(\phi)\cap perms(\psi)$ then one of those seeds (denote it by $\phi$) is equal to $(p_1,p_2,...,p_n)$ (without an element $p_j=p_2\oplus 1$) and the other
 (denote it by $\psi$) to $(p_1,p_3,...,p_n)$.
 Removing both elements from $\pi$ gives us the sequence obtained from $\phi$ and $\psi$ after the appropriate removals.
 
 Permutations 
$\widetilde{\psi}=\pi$ and $\phi^{(n-j+1)}=(p_j,p_{j+1},...,p_n,p_1,...,p_{j-1})$ 
are cyclically equivalent.
 If $height(\phi)=k<n-3$ and $j>k+1$ then elements $p_2,p_3,...,p_{k+1}$ 
form a sequence decreasing by one: a sequence with a property $p_i=p_{i+1}\oplus 1$ and
 $p_{k+1}\neq p_{k+2}\oplus 1$ ($p_{k+1}\neq p_{k+3}\oplus 1$ if $j=k+2$). 

We obtain seed $\psi$ by removing the element $p_2$ from $\phi$ and inserting the element $p_j=p_2\oplus 1$
 (between elements $p_{j-1}$ and $p_{j+1}$). The removal shortens the decreasing sequence by 1 
and the insertion can neither shorten it (as the element lands outside of the 
sequence) nor extend it
 (even if $j=k+2$ as for $p_j=p_{k+1}\ominus 1$ the pair $(k,j)$ would have to be 
equal to $(n-2,n)$ which is impossible outside of $Hub_n$). Hence 
$$height(\psi)=k-1=\min(k-1,j-3)=\Delta(k,n-j+1)=\Delta(k,i)$$
 If $j\le k+1$ then the decreasing sequence is formed by elements $p_2,p_3,...,p_{k+2}$ (omitting element $p_j$). The removal of $p_2$ decreases the length of the sequence by 1 and insertion of $p_j$
 cut it just before that element. We have 
$$height(\psi)=j-3=\min(k-1,j-3)=\Delta(k,n-j+1)=\Delta(k,i)$$
 We can remove $p_2$ from the seed $\phi$ and insert $p_2\oplus 1$ in any of the $n-3$ places after $p_3$ obtaining a valid seed which is a son of $\phi$.
 
 Hence we know that packages of two seeds intersects if and only if they are in a parent-son relation, and that every seed of height greater than 1 has exactly $n-3$ seed-sons.
 We can also easily check which son of $parent(\psi)$ the seed $\psi$ is or compute $son(\psi,i)$ for any $i\in\{1,...,n-3\}$.
\end{proof}

\medskip
\begin{figure}[h]
\vspace*{-0.5cm}
\centerline{\includegraphics[width=11cm]{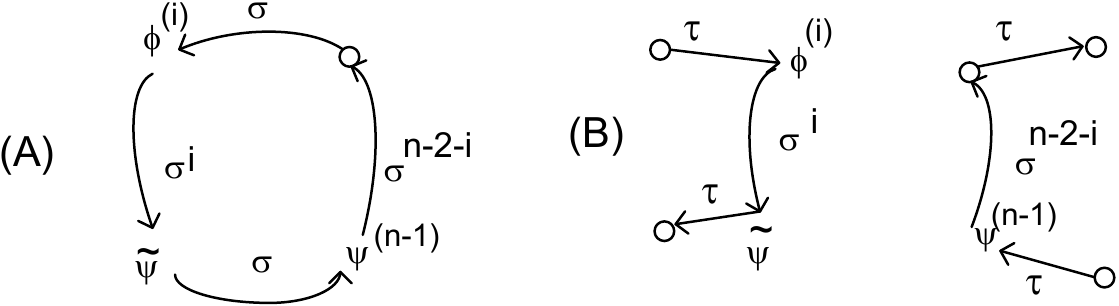}}
\caption{(A) The anatomy of $perms(\phi)\cap perms(\psi)$:
the graph $\SeedGraph(\psi)\cap \SeedGraph(\phi)$.\ (B) A part of the Hamiltonian path $\PATH(n)$
after removing two conflicting $\sigma$-edges, we have that 
$\psi$ is the $i$-th son of $\phi$.}\label{Sept26}
\end{figure}

For $k<n-3$  and a seed $\psi$ of height $k$ we define $\W_k$ as the sequence of labels 
of a sub-path in $\PATH(n)$  starting in $\widetilde{\psi}$ and
ending in $\psi^{(n-1)}$.
In other words it is a 
$\sigma\/\tau$-sequence
generating all $n$-permutations (each exactly once) 
of $bunch(\psi)$.

\begin{observation}
By Lemma \ref{seed intersections} every seed $\psi$ such that $1<height(\psi)<n-3$ 
has exactly $n-3$ sons whose heights depend only on height of  $\psi$. 
Hence (by induction on heights) all trees $Tree(\psi)$ are
isomorphic for seeds $\psi$ of the same height. Consequently the 
definition of $\W_k$ is justified as it depends only on the height of $\psi$.
\end{observation}

For a permutation $\pi$ and a sequence $\alpha$ of operations $\sigma,\tau$
denote by $GEN(\pi,\alpha)$ the set of all permutations generated from $\pi$
by following $\alpha$, including $\pi$.

\noindent The word $\W_k$ satisfies: 

$GEN(\widetilde{\psi},\, \W_k)\,=\, bunch(\psi)\ \ \mbox{and}\ \ 
\W_k(\widetilde{\psi})=\psi^{(n-1)}.$

\noindent In this section we give compact representation of $\W_k$

\noindent
For example if $height(\psi)=1$ then $W_1$ is a traversal of $perms(\psi)$
except $n-2$ cyclically equivalent permutations, common to
$perms(\psi)$ and $perms(\phi)$, where $\phi=parent(\psi)$.

\medskip
\begin{figure}[ht]
\centerline{\includegraphics[width=12.5cm]{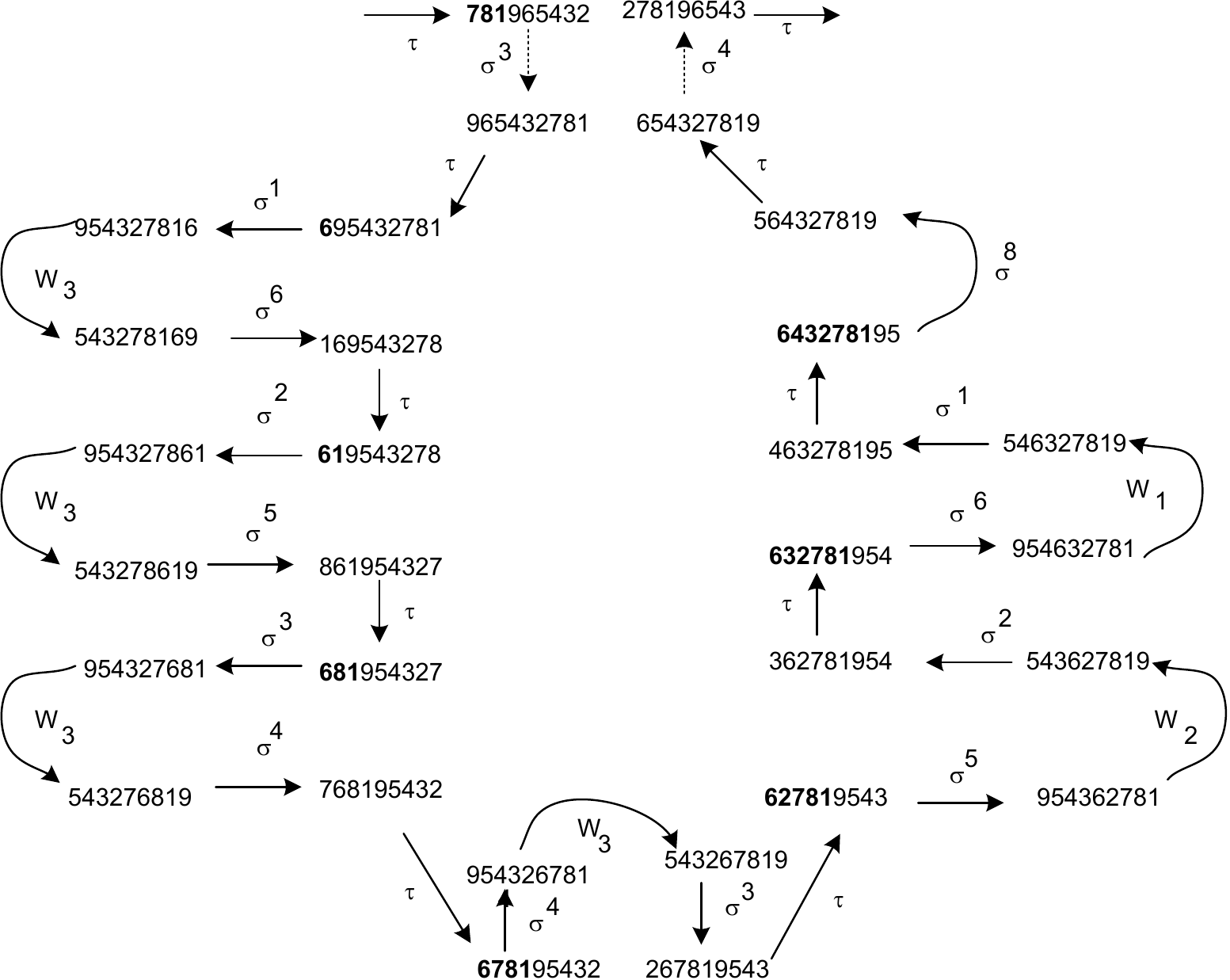}}
\vspace*{0.2cm}
\caption{The structure of $bunch(\psi)$ for the seed $\psi=95432781$.
We have $parent(\psi)=\phi$, where $\phi=96543281$.
The connecting points of $\psi$ with its parent 
are $\widetilde{\psi}$ and $\psi^{(n-1)}$, in other words $bunch(\psi)
\cap perms(\phi)\;=\; \{\widetilde{\psi},\,\psi^{(n-1)}\}$. 
The sequence $W_4$ starts in $\widetilde{\psi}$,
visits all permutations in $bunch(\psi)$ and ends 
in $\psi^{(n-1)}$. We have:\vspace*{0.1cm}
$\W_4\,=\,\tau\, \cdot\, 
 \sigma^1 \W_3  \gamma_{6}   \cdot  \sigma^2 \W_3  \gamma_5   
\cdot  \sigma^3 \W_3 \gamma_4   \cdot  \sigma^4 \W_3 \gamma_3  
\cdot   \sigma^5\W_2 \gamma_2    \cdot  
\sigma^6 \W_1 \gamma_1 \, \cdot\,  \gamma_{8}$
}\label{fig: example1}
\end{figure}

\noindent Recall that we denote $\gamma_k\,=\, \sigma^k\tau$
\begin{theorem}\label{recurrences for W_k}
For $1\le k < n-3$ we have the following recurrences: 

\medskip
{$\W_0=\sigma,\ \ \ \W_k \;=\; \tau\ \cdot\ \prod_{i=1}^{n-2}\,\sigma^i\, \W_{\Delta(k,i)}\,\gamma_{n-2-i}$}
\end{theorem}
%-------------------
\begin{proof}
Assume $\psi\notin Hub_n$ is of height $k$, then by Lemma~\ref{seed intersections} the first,
from left to right,
$n-k-1$ children of $\psi$ in the subtree $Tree(\psi)$ 
are of height $k-1$ and the next $k-2$
children are of heights $k-2,k-3,...,1$.
The representative $\widetilde{\beta_i}$ of the $i$-th son $\beta_i$ of $\psi$ 
equals $\sigma^i(\psi^{(i)})$ (see Figures~\ref{fig: example1} and~\ref{fig: example2} ).
 \end{proof}

\medskip
\begin{figure}[h!]
\centerline{\includegraphics[width=9cm]{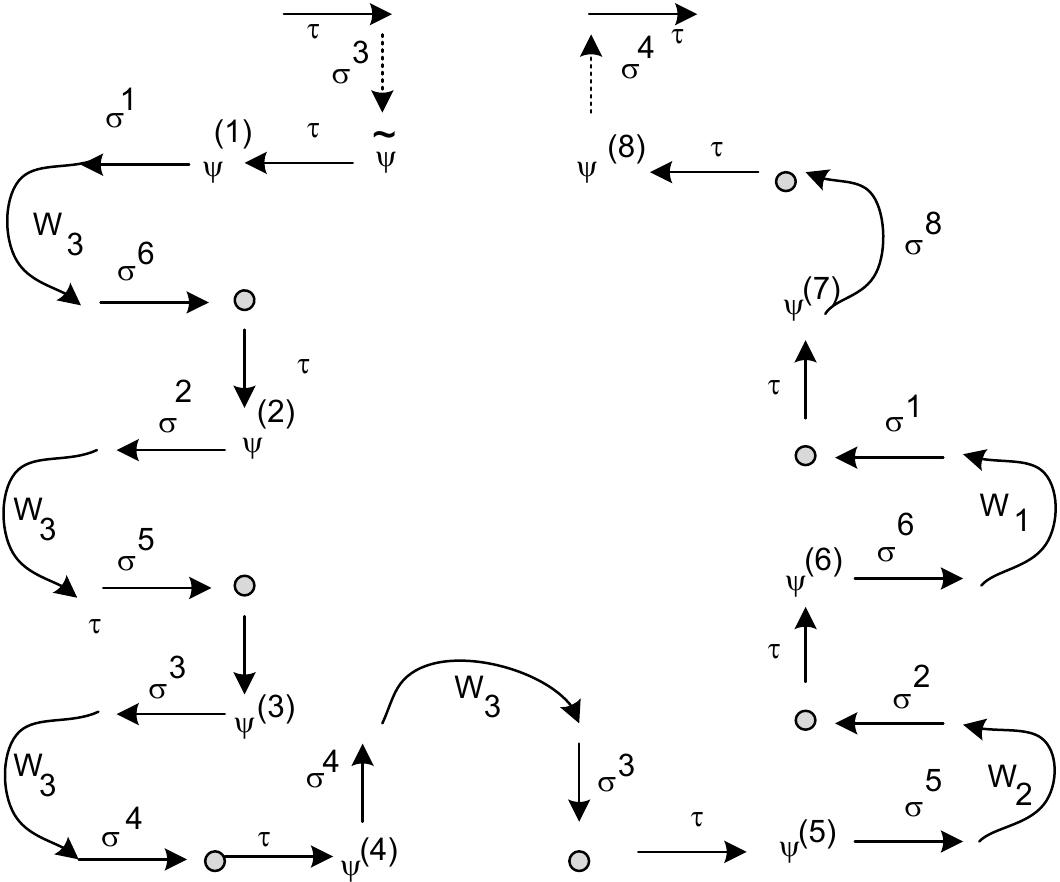}}
\caption{Schematic view of structure from Figure \ref{fig: example1}.}\label{fig: example2}
\end{figure}

\section{Compact representation of the whole generation}
We have the following fact: 
\begin{observation}
Assume two seeds $\psi,\beta$ satisfy:
$height(\psi)=k>1$ and\\
$\sigma^i(\psi^{(i)})=\widetilde{\beta}$. Then
if $i=1$ and $ \psi\in Hub_n$ then $height(\beta)=height(\psi)$.
\end{observation}
\begin{theorem}
The whole $\sigma\tau$-sequence $\Seq_n$ starting at $\tau(n,n-1,...,1)$,
ending at $\sigma\tau(n,n-1,...,1)$, and  generating all
$n$-permutations, has the following compact representation
of $\Oh(n^2)$ size (together with recurrences for $\W_k$):

\medskip\centerline{$ \Seq_n\;=\; 
\gamma_1^{n-2}\sigma^2\;(\V_n\,\tau)^{n-2}\; \V_n,\ \ \mbox{\rm where}
$}

\medskip\centerline{$ \V_n \;=\; \gamma_{n-3}\ \cdot\ \prod_{i=2}^{n-3}\,\sigma^i\,
\W_{\Delta(n-3,i)}\,\gamma_{n-2-i}\ \cdot \ \sigma^{n-1}.$}
\end{theorem}
\begin{proof}
For every non-hub seed $\psi$ we had that
 $GEN(\widetilde{\psi},\, \W_k)\,=\, bunch(\psi)$, where $k=height(\psi)$.
\noindent The only difference for a hub seed $\phi$ is that $son(\phi,1)$ cannot be considered
as part of a tree rooted at $\phi$ (with already defined
$parent$-links), since $son(\phi,1)\in Hub_n$ and this would lead to a cycle
($son(\phi,1)$ is reachable via parent-links from $\phi$).
\noindent Thus to prevent this problem we define $V_n$ as $W_{n-3}$ with
the part corresponding to the first son removed (leaving only the $\gamma_{n-2-1}$ part),
 and also delete the last symbol $\tau$, as it does not appear at the end of the path
(it corresponds to one of the $\tau$-edges removed when joining two cycles into one path).
Now $Seq_n$ consists of $n-1$ such segments $V_n$
(corresponding to $n-1$ hub seeds) joined by $\tau$-edges
(they are linked in the same way as if the previous $V_n$ part was a son of the next one).
Additionally it starts with $\gamma_1^{n-1}$-path representing the small path
with the last $\tau$-edge replaced by a $\sigma$-edge.
\end{proof}
\vspace*{-.3cm}

\medskip
\begin{figure}[h]
\centerline{\includegraphics[width=12.5cm]{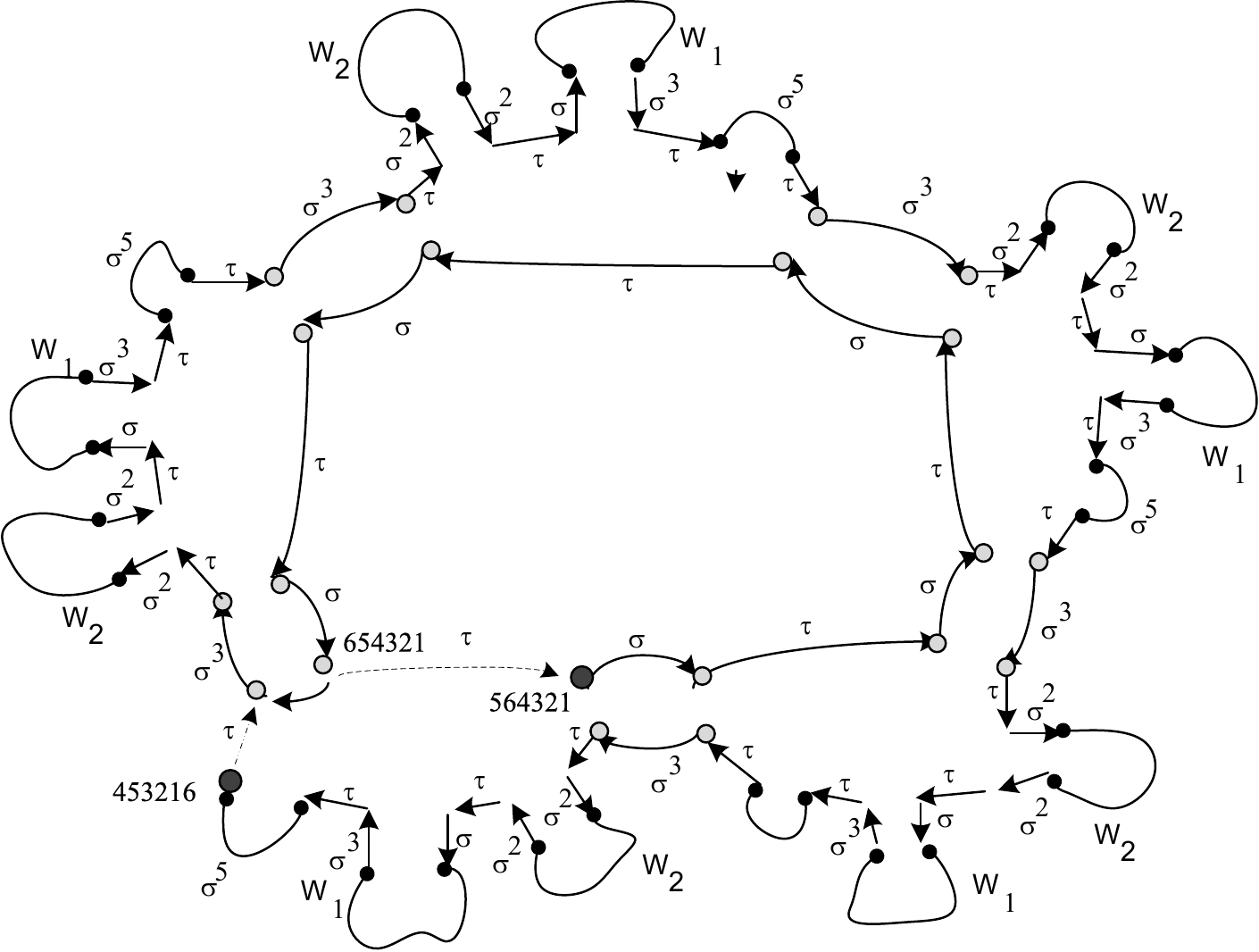}}
\caption{The  compacted structure of 
$\Seq_6$ of length 720. It differs from  the structure of $\R_6$ by 
adding one $\sigma$-edge from $654321$ and removing two (dotted) $\tau$-edges
to have Hamiltonian path.  
We have: 
$\Seq_6\;=\; (\sigma\tau)^4\sigma^2\;(\V_6\tau)^4\;\V_6$, where $\V_6\;=\;
\sigma^3\tau\; \sigma^2\W_2\sigma^2\tau\; \sigma \W_1\sigma^3\tau\; \sigma^5$. 
The structure is the union of graphs of 5 seeds in $Hub_6$ with {\it hanging} bunches.
The {\it starting path} consists of permutations from $564321$ to $654321$.
}\label{fig: example4}
\end{figure}

%
%\vspace*{-1cm}
\section{Ranking}
We need some preprocessing to access later some values in constant time.
\begin{observation}\label{sizes of trees}
 All the values $|W_k|$ and $\sum\limits_{i=0}^k(|W_i|+n-1)$ 
for $k\in\{0..n-4\}$ can be computed 
in $\Oh(n)$ total time and accessed in $\Oh(1)$ time afterwards.
\end{observation}
\noindent
The ranks of representatives of hub seeds are easy to compute. For example
for $n=6$ we have (see Figure~\ref{fig: example4}):
 $rank(643215)=1,\;rank(632154)=3,$ $rank(621543)=5,\;
rank(615432)=7,\;rank(654321)=9.$
\begin{lemma}\label{2cases - ranking}
 For a given permutation $\pi$ we can compute
in time $\Oh(n)$ 
\\
{\bf (a)} $rank(\pi)-rank(\widetilde{\psi})$ if $\pi\in perms(\psi)$, 
\\
{\bf (b)} $rank(\pi)$ if $\pi\in perms(\psi)$ for some $\psi \in Hub_n$. 
\end{lemma}

\begin{proof}\mbox{ \ }\\
By the {\it starting path} we mean the sequence on the first $2n-2$ permutations of $PATH(n)$, see Figure \ref{fig: example4}.

\smallskip\noindent
{\bf (a)}
 Given a permutation $\pi\in perms(\psi)$ we define $\rho=(n,r_2,...,r_n)$ as 
permutation cyclically 
equivalent to $\pi$ and starting with $n$, and let $j$ be the position 
such that $r_j=r_2\oplus 1$.
If $\psi=(n,r_3,...,r_n)$, then $\rho=\widetilde{\psi}$

\smallskip\noindent
 Now we distinguish two cases depending on the position 
$l$ of $n$ in $\pi$. 

\begin{description}
\item{Case 1:}
if $l\le n-j+2$, then $rank(\pi)=rank(\rho)-l+1$

 (it appears in the $\sigma_{n-j+1}$ part of $bunch(parent(\psi))$), 

\item{Case 2:}
otherwise 
$rank(\pi)=|W_{height(\psi)}|+(n-1)-l+1$ 

(it appears in the $\gamma_{j-3}$ part).
\end{description}

\smallskip\noindent
 If $\psi=(n,r_2,...,r_{j-1},r_{j+1},...,r_n)$, then 
$$rank(\rho)=rank(\widetilde{\psi})+SUM(height(\psi),n-j+1)$$
 Using similar arguments as before we have that:

 $rank(\pi)=rank(\rho)-l+1$ if $l\le n-j+2$ and 

$rank(\pi)=rank(\rho)+|W_{\Delta(height(\psi),n-j+1)}|+(n-1)-l+1$ otherwise.

\medskip\noindent
{\bf (b)}
 If the permutation $\pi$ after removing $n$ is cyclically equivalent to $(n-1,n-2,...,1)$ and $n$ appears on the first position ($\pi$ belongs to the starting path)
 then $rank(\pi)=2\cdot(n-p_2-1)-1$, and if it appears on the second position then $rank(\pi)=2\cdot(n-p_1-1)$, where $p_1,p_2$ are the first two positions of $\pi$.\\
 Otherwise we define $\rho$, $j$ and $l$ like in case (a) and $\psi=(n,r_2,...,r_{j-1},r_{j+1},...,r_n)$. We know that $\pi\in perms(\psi)$, and want to compute $rank(\pi)$ minus the rank of the first permutation of
 $\psi$ which appears in $PATH(n)$ with rank greater than $2n-3$.
 That permutation is equal to $$\mu=(r_2,...,r_{j-1},r_{j+1},...,r_n,r_2\oplus 1,n)$$
 (equal to $\sigma^2(\psi^{1})$ if $\psi^{(1)}$ was defined for hub seeds in the same way as for the other ones),
 and has rank equal to $$(2n-2)+|V_n\tau|\cdot(r_2 mod (n-1))\;=\;
2n-2+(r_2 mod (n-1))\cdot(n(n-2)!-2)$$

\smallskip\noindent
 If $j=n$ then $rank(\psi)=rank(\mu)+n-l$ else if $l\le n-j+2$ then 
 $$rank(\rho)=rank(\mu)+SUM(n-3,n-j+1)-|W_{n-4}|-2, \  
rank(\pi)=rank(\rho)-l+1$$ Otherwise we have
 $$rank(\pi)=rank(\rho)+|W_{j-3}|+(n-1)-l+1.$$
 
 Those two algorithms lets us rank permutations in basic 
cases, and allows us to reduce the main problem to a simpler one (ranking the representatives of seeds).
\end{proof}
Hence we concentrate on ranking permutations of type $\widetilde{\psi}$ (representatives of seeds).
We slightly abuse notation and for a seed $\psi$ define $rank(\psi)=rank(\widetilde{\psi})$.

For a non-hub seed $\psi$ denote by $anchor(\psi)$ the highest non-hub ancestor $\phi$ of $\psi$ and let $hub(\psi)=parent(anchor(\psi))$.
Observe that the  anchor $\phi$ is the first contacting seed with the hub, it is the first
ancestor of $\psi$ such that
 $perms(\phi)\cap perms(\beta)\ne \emptyset$ for some $\beta\in Hub_n$, in fact for $\beta=hub(\psi)$.

\smallskip\noindent
$rank(\psi)-rank(anchor(\psi))$ for a non-hub seed $\psi$, can be treated as
its {\it distance from $Hub_n$}.
\noindent It happens that computing the rank of the anchor is much easier, since we have to deal only with 
the hub seeds.
\noindent The bottleneck in ranking is computation of the 
distance of a seed representative from $Hub_n$.
Define:  

\medskip
\centerline{$SUM(k,j) \;=\; |\tau\ \cdot\ \prod_{i=1}^{j-1}\,\sigma^i\,
\W_{\Delta(k,i)}\,\gamma_{n-2-i}|\;+\; j.$}

\smallskip
\noindent Denote also by $ord(\psi)$ the position of $mis(\psi)+1$ in $\psi$ counting from the end of sequence $\psi$.
For example for $\psi=(10\,6\,5\,9\,8\,4\,3\,1\,2)$ we have $ord(\psi)\,=\, 5$, since $mis(\psi)=7$
and 8 is on the 5-th position from the right.

\begin{observation}\label{Oct23}\mbox{ \ }
If $\phi=parent(\psi)\notin Hub_n$ and $\psi$ is the $i$-th son of $\phi$
then 
$rank(\psi)-rank(\phi)\;=\; SUM(height(\phi),\, i).$
\end{observation}

\begin{example} Let $\psi=94326781$, then $parent(\psi)= \phi= 95432781$.
The path from $\widetilde{\phi}=965432781$ to $\widetilde{\psi} =  954326781$ is
$$
\tau\,\sigma^1W_3\sigma^6\tau\, \sigma^2W_3\sigma^5\tau\, \sigma^3W_3\sigma^4\tau\,  \sigma^4,$$
see Figure~\ref{fig: example1}.
Its length equals $SUM(4,4)$, we have: 
$height(\phi)=4, \; ord(\psi)=4.$ 
\end{example}

\begin{observation}
$ord(\psi)=i$ iff $\psi$ is the $i$-th son of $parent(\psi)$.
\end{observation}

\noindent 
For the parent-sequence $\psi_0=\psi, \psi_1, ..., \psi_m=anchor(\psi_i)$
denote

\smallskip\centerline{$route(\psi)=(ord(\psi_0), ord(\psi_1), ... ,
ord(\psi_m)).$}

\smallskip\noindent
For a seed $\psi=a_1a_2...a_{n-1}$ 
 define the decreasing sequence of $\psi$, denoted by $dec\_seq(\psi)$, 
as the maximal sequence $a_{i_0}a_{i_1}...a_{i_m}$, where
$2=i_0<i_1<i_2<...<i_m$ such that $i_{j-1}=i_j\oplus 1$ for $0<j\leq m$.
Denote $level(\psi)=n-m-3$.
The length of the {\it parent-sequence} $\psi=\psi_0,\psi_1,\psi_2,...,\psi_r=anchor(\psi)$
from $\psi$ to its anchor
is $r=level(\psi)-1$.
\begin{example}\label{Oct28} We have: $dec\_seq(96154238)=(6,5,4,3)$.
Hence the path from $\psi=(96154238)$ to $anchor(\psi)=(98765423)$ is
of length $(9-3-3)-1 = 2$. This path equals:

\smallskip\centerline{$\psi_0\rightarrow \psi_1\rightarrow \psi_2\ 
=\ 96154238\rightarrow 97615423\rightarrow 98765423.$}

\smallskip\noindent  
We have: $ord(\psi_0)=1,\ ord(\psi_1)=5, \ ord(\psi_2)=2$,
 $route(96154238)\,=\, (1,\,5,\,2)$.
\end{example}

\noindent
The key point is that we do not need to deal with
the whole parent-sequence, including explicitly seeds on the path, which is
of quadratic size (in worst-case) but
it is sufficient to deal with the sequence of orders of sons, 
which is an implicit representation of this path of only linear size
\begin{lemma}\label{ord}
For a non-hub seed $\psi$ we can compute $route(\psi)$ and  $anchor(\psi)$
in $\Oh(n\sqrt{\log n})$ time.
\end{lemma} 
\begin{proof} We know the length of the parent sequence from $\psi$ to its anchor,
since we know  $level(\psi)$.
 Now we use the following auxiliary problem 

\smallskip
{\bf Inversion-Vector problem}:

 \hspace*{0.5cm} for a seed $\psi$ compute for each element $x$ 
the number  $RightSm[x]$

\hspace*{.5cm}
of  elements smaller than $x$ which are to the right of $x$ in $\psi$.

\medskip\noindent Assume  $\psi=(a_1,a_2,...,a_{n-1})$.
 We introduce a new linear order  

\smallskip\centerline{$a_2\prec a_2\ominus 1\prec a_2\ominus 2\prec ... \prec a_2\ominus (n-2)$.}

\smallskip\noindent
Then we compute together the numbers $RightSm[z]$ 
w.r.t. linear order $\prec$  for each element $z$ in $\psi$.

\smallskip\noindent Now $ord(\psi_i)$ is computed separately
for each $i$ in the following way:

\smallskip\centerline{$ord(\psi_i)\,:=\,RightSm[x_i+1]+1$,
  where $x_i=mis(\psi_i)$ }

\medskip\noindent
The {\it Inversion-Vector} problem can be computed in 
 $\Oh(n\sqrt{\log n})$ time, see \cite{DBLP:conf/soda/ChanP10}.
Consequently the whole computation
of numbers $ord(\psi_i)$ is of the same asymptotic complexity. 
We know that $hub(\psi)\;=\;(n,b,b\ominus 1,...,b\ominus (n-3))$, where ${b=a_2\oplus level(\psi)}$ and we know also which son of $hub(\psi)$ is $anchor(\psi)$.
This knowledge allows to compute $anchor(\psi)$ within required complexity.
This completes the proof.
\end{proof}

\smallskip\begin{corollary}\label{anchor-rank}
For a non-hub seed $\psi$ the value $rank(\psi)-rank(anchor(\psi))$ can be 
computed in $\Oh(n\sqrt{\log n})$ time.
\end{corollary}
\begin{proof}
Let the {\it parent-sequence} from $\psi$ to its anchor
be 

\centerline{$\psi=\psi_0,\psi_1,\psi_2,...,\psi_r=anchor(\psi),$ where $r=level(\psi)-1$.}

\smallskip\noindent Then 
$rank(\psi_{i})-rank(\psi_{i+1})=SUM(height(\psi_{i+1}),ord(\psi_i)),$
  and

\smallskip\noindent $height(\psi_i)=\Delta(height(\psi_{i+1}),ord(\psi_i))$,
which allows us to compute in $\Oh(n)$ time:

\centerline{$rank(\psi)-rank(anchor(\psi))=\sum_{i=m-1}^{0}\,(rank(\psi_{i})-rank(\psi_{i+1}))$} 

\smallskip
Now the thesis is a 
consequence of Observation~\ref{sizes of trees}, Observation~\ref{Oct23} and Lemma~\ref{ord}.
This completes the proof.
\end{proof}

\begin{example} (Continuation of Example~\ref{Oct28}) For $\psi$ from Example~\ref{Oct28} we have: 

\smallskip
\centerline{$rank(\psi)-rank(anchor(\psi))\;=\; SUM(5,5)+SUM(2,1)$}
\end{example} 

\begin{comment}
------------------------

\end{comment}
\noindent The following result follows directly from Corollary~\ref{anchor-rank}, Lemma~\ref{2cases - ranking} and Observation~\ref{sizes of trees}.
\begin{theorem}{\rm\bf [Ranking]}
 For a given permutation $\pi$ we can compute the rank of $\pi$ in $\Seq_n$ 
in time $\Oh(n\sqrt{\log n})$
\end{theorem}
%
%------------------------
\section{Unranking}
Denote by $Perm(t)$ the $t$-th permutation in $SEQ_n$, and for $t< |bunch(\psi)|$  let 
$Perm(\psi,t)=Perm(t+rank(\widetilde{\psi}))$ (it is the $t$-th permutation in $bunch(\psi)$,
counting from the beginning of this bunch).
The following case is an easy one.
\begin{lemma}\label{1case - unranking} If we know a seed $\psi$ together with its rank, such that\\ 
$Perm(t) \in perms(\psi)$, then we can recover $Perm(t)$ in linear time.
\end{lemma}

\begin{proof}\mbox{ \ }\\
 Let $\psi=(n,a_2,...,a_{n-1})$, and $k=height(\psi)$.
 In linear time we find $j$ such that $SUM(k,j)-j\le t<SUM(k,j+1)-j-1$.
\\
 If $SUM(k,j)<t<SUM(k,j)+|W_{\Delta(k,j)}|$, then $Perm(\psi,t)$ does not belong to $perms(\psi)$ (it belongs to $bunch(son(\psi,j))$).
 \\ If $l=SUM(k,j)-t\ge 0$, then $Perm(\psi,t)$ is equal to $(n,a_2,...,a_{n-j},a_2\oplus 1,a_{n-j+1})$ rotated by $l$ to the right,
 and if $l=t-SUM(k,j)+|W_{\Delta(k,j)}|\ge 0$, then it is the same permutation rotated by $l+1$ to the left.
 
 In this way we reduced the problem of unranking permutation outside of $Hub_n$ to finding a package containing the permutation.
\end{proof}

We say that a permutation $\pi$ is a {\it hub-permutation} if
$\pi\in perms(\psi)$ for some $\psi\in Hub_n$.

\begin{lemma}\label{2case - unranking}
 We can test in $\Oh(n)$ time if $Perm(t)$
is a hub-permutation.\\
{\bf (a)}If ''yes'' then we can recover $Perm(t)$ in $\Oh(n)$ time.\\
{\bf (b)} Otherwise we can find in $\Oh(n)$ time
an anchor-seed $\psi$ together with $rank(\psi)$ such that
$Perm(t)\in bunch(\psi)$.
\end{lemma}

\begin{proof}\mbox{ \ }\\

 If $t<2n-2$ then $Perm(t)$ is equal to $(n-1,...,1)$ rotated to the left by $\lceil\frac{t}{2}\rceil$, with $n$ inserted on first position if $x$ is odd, and on the second if it is even.
\\
 Otherwise let $t-(2n-2)=t_1\cdot|V_n\tau|+t_2$ (we use integer division), and let 
$$\psi=(n,t_1,t_1\ominus 1,...,t_1\ominus (n-3))$$ (with $t_1$ substituted by $n-1$ if equal to $0$).
 $Perm(t)$ belongs to $perms(\psi)$, or to $bunch(\phi)$, where $parent(\phi)=\psi$. 

If $t_2<n-2$ then $Perm(t)$ equals to $(n,t_1,t_1\ominus1,...,t_1\ominus(n-2))$ rotated to the left by $t_2+1$.
 
\smallskip
In the other case in linear time we find $j$ such that 
$$SUM(n-3,j)-j\le t_2+(1+|W_{n-4}|+n-1)-(n-2)<SUM(n-3,j+1)-j-1,$$
$$\mbox{and}\ l=t_2+|W_{n-4}|+2-SUM(n-3,j),\
\phi=son(\psi,j)$$
 If $l\le0$ then $Perm(t)$ is equal to $\widetilde{\phi}$ rotated to right by $-l$, else if $l\ge|W_{\Delta(n-3,j)}|$, then $Perm(t)$ is equal to $\widetilde{\phi}$ rotated to the left by
 $l-|W_{\Delta(n-3,j)}|+1$. 
Otherwise $Perm(t)=Perm(\phi,l)$, and it is not a hub-permutation.

By using this algorithm we either already succeed in finding the right permutation, or restrict ourselves to a limited regular part of $PATH(n)$.
\end{proof}

\noindent For a sequence $\mathbf{b}=(b_1,b_2,...,b_m)$ of positive integers
denote 

\smallskip\centerline{$MaxFrac(\mathbf{b})\,=\, \max_i\;\frac{b_{i+1}}{b_i},\ \
MinFrac(\mathbf{b})\,=\, \min_i\;\frac{b_{i+1}}{b_i}.$}

\medskip\noindent
The sequence $\mathbf{b}$ is called here   $D(m)$--{\it stably increasing}
iff 

\smallskip\centerline{
$MinFrac(\mathbf{b})  \ge 2$, and  $MaxFrac(\mathbf{b}) \le D(m)$.}

\begin{lemma}\label{fast search & sums of W_i are stably increasing}\mbox{ \ }\\
{\bf (a)}
 Assume we have a $D(m)$--stably increasing sequence $\mathbf{b}$ of length $\Oh(m)$.
 Then after linear preprocessing we can locate any integer $t$ in the sequence $\mathbf{b}$ 
in \\ $\Oh(\log\log D(m))$ time.\\
{\bf (b)}
 The  sequence $\mathbf{b}\;=\; (b_0,b_1,...,b_{n-5})$,
 where $b_k=\sum_{i=0}^{k}(|W_i|+n-1)$ is $n$--stably increasing.
\end{lemma}

\begin{proof}\mbox{ \ }\\
{\bf (a)}
 Denote $B=\frac{b_m}{b_1},\delta=\sqrt[m]{B,}$ $Min=MinFrac(\mathbf{b}),Max=MaxFrac(\mathbf{b})$,
 and let $d$ be a sequence such that $d[i]=\max\{j:b_j\ge b_1\cdot \delta^i\}$.

 It can be computed in linear time by scanning the sequence $\mathbf{b}$ from left to right and reporting whenever element exceeds next power of $\delta$.
 
 Thanks to the fact that sequence $\mathbf{b}$ is $D(m)$--stably increasing the maximal difference between consecutive values of a sequence $d$ is bounded by
 $\log_{Min}(\delta)\le\log_{Min}(Max)\le\log_2(D(m))$.

\smallskip\noindent 
 When locating $x$ in the sequence $\mathbf{b}$ we can compute the value 
$$y=\lfloor\log_{\delta}{\left(\frac{x}{b_1}\right)}\rfloor
 =\lfloor\frac{m\cdot(\log_2{x}-\log_2{b_1})}{\log_2{B}}\rfloor$$
 It now suffices to binary scan the part of sequence $\mathbf{b}$ between positions $d[y]$ and $d[y+1]$, which has a length bounded by $\log_2(D(m))$.\\

\medskip\noindent
{\bf (b)}
For any $k\in \{1,...,n-5\}$
the value of $b_k=\sum_{i=0}^{k-1}(n-1+|W_i|)+n-1+|W_k|$ equals 
$$2\cdot\sum_{i=0}^{k-1}(n-1+|W_i|)+(n-k-2)\cdot(|W_{k-1}|+n-1)+n$$
$$>2\cdot\sum_{i=0}^{k-1}(n-1+|W_i|)=2\cdot b_{k-1}$$
\noindent We have:
$$\frac{b_k}{b_{k-1}}= 2+\frac{n+(n-k-2)\cdot(|W_{k-1}|+n-1)}{\sum_{i=0}^{k-1}(n-1+|W_i|)}$$
$$<2+\frac{n+(n-k-2)\cdot(|W_{k-1}|+n-1)}{(n-1+|W_{k-1}|)}$$
$$=\;n-k+\frac{n}{(|W_{k-1}|+n-1)}\le n-k+1\le n$$ $ $

Combining those two parts of the lemma allows us to locate values in the sequence $b_k=\sum_{i=0}^{k}(n-1+|W_i|)$ in $\Oh(\log\log n)$
time after linear preprocessing dependent only on the value of $n$.
\end{proof}

\begin{lemma}\label{fast search in tree}
After linear preprocessing if we are given a height of a non-hub seed $\psi$, and a number $t\le|bunch(\psi)|$ we can
find the number $j$ and $height(\beta)$ of the seed-son $\beta$ of $\psi$ such that $Perm(\psi,t)\in bunch(\beta)$
in $\Oh(\log\log n)$ time
if $Perm(\psi,t)\notin perms(\psi)$.
\end{lemma} 
\begin{proof}
Let $k=height(\psi)$.
We need  $j$ such that $SUM(k,j)-j\le t<SUM(k,j+1)-(j+1)$.
For $j\le n-k$ we have $SUM(k,j)-j=(j-1)\cdot(|W_{k-1}|+n-1)$, hence if
$t<SUM(k,n-k)-n+k$ the simple division by $|W_{k-1}|+n-1$ suffices to find the appropriate $j$.
Otherwise we look for $j$ such that 
$$|W_k|-SUM(k,j+1)+j+1< s\le |W_k|-SUM(k,j)+j,\ \mbox{where}\ s=|W_k|-t.$$

\noindent
Let
{$b_i=|W_k|-SUM(k,n-2-i)+n-2-i=(\sum_{j=0}^i |W_j|+n-1).$}

\smallskip
By Lemma \ref{fast search & sums of W_i are stably increasing}(b) $(b_0,...,b_{k-2})$ is $n$--stably increasing 
(it is a prefix of $(b_0,...,b_{n-5})$ for which we made the linear preprocessing).
Hence by Lemma \ref{fast search & sums of W_i are stably increasing}(a) we can find the required 
$j$ in $\Oh(\log\log n)$ time.\\
Moreover if $SUM(k,j)<t<SUM(k,j)+|W_{\Delta(k,j)}|$, then $Perm(\psi,t)=Perm(\beta,t-SUM(k,j))$, where $\beta=son(\psi,j)$ has height $\Delta(k,j)$.
Otherwise $Perm(\psi,t)\in perms(\psi)$.
\end{proof}

\begin{theorem}{\rm\bf [Unranking]}
 For a given number $t$ we can compute the $t$-th permutation in
Sawada-Williams generation in $\Oh(n\frac{\log n}{\log\log n})$.
\end{theorem}
\begin{proof}
 From Lemma \ref{2case - unranking} we either obtain the required 
permutation (if it is a hub-permutation) or obtain its anchor-seed $\phi$
and $rank(\phi)$.
 In the second case we know that $Perm(t)\in bunch(\phi)$ and it equals 
$Perm(\phi,\,t-rank(\widetilde{\phi}))$.
Now after the linear preprocessing we apply Lemma \ref{fast search 
in tree} exhaustively to obtain $route(\psi)$ for a seed $\psi$ 
such that $Perm(t)\in perms(\psi)$. However we do not know $\psi$ and have to compute it.
\begin{claim}\label{Nov15}
If we know $anchor(\psi)$ and $route(\psi)$ then $\psi$ can be computed
in \\ $\Oh(n\frac{\log n}{\log\log n})$ time. 
\end{claim}
\begin{proof}
We can compute the second element $a_2$ of $\psi$ as $a'_2\ominus m$ and
$dec\_seq(\psi)$ as $(a_2,a_2\ominus1,...,a_2\ominus(n-m-3))$
where $a'_2$ is the second element of $anchor(\psi)$, and $m=|route(\psi)|-1$.
Then we use the order:

\centerline{$a_2\prec a_2\ominus 1\prec a_2\ominus 2\prec ...\prec a_2\ominus (n-2)$.}

\smallskip\noindent 
We produce a linked list initialized with $dec\_seq(\psi)$.
For $i\in\{0,...,m-1\}$ 
we want to insert $a_2\oplus (m+1-i)$ after $ord(\psi_{m-1})$ position from the end of the current list
(all the smaller elements are already in the list and we know, that after $a_2\oplus(m+1-i)$ there are $ord(\psi_{m-1})-1$ such elements).
$\psi$ is composed of $n$ and consecutive elements of the final list.
The data structure from \cite{DBLP:conf/wads/Dietz89} 
allows us to achieve that in $\Oh(n\frac{\log n}{\log\log n})$ time.
\end{proof}
 Finally we use this claim and Lemma~\ref{1case - unranking}
 to obtain the required permutation $Perm(t)$.
\end{proof}
\section{Examples of ranking and unranking }
\noindent We show on representative examples
how the ranking and unranking algorithms are working.
\paragraph{\bf Ranking.} 
When ranking  $\pi=(7,2,4,1,6,5,10,9,8,3)$ we first find a permutation $\rho=(10,9,8,3,7,2,4,1,6,5)$ cyclically equivalent to $\pi$ and then
a seed $\psi=(10,9,8,3,7,2,4,6,5)$ whose 
package $perms(\psi)$ contains both $\pi$ and $\rho$ (in case of two candidates for $\psi$ we choose the 
parent).
We have $rank(\pi)-rank(\widetilde{\psi})=(rank(\pi)-rank(\rho))+(rank(\rho)-rank(\widetilde{\psi}))=(|W_1|+3)+(SUM(2,3))=268$. 
Next  we compute $route(\psi)$ by computing inversion vector of $\psi$.
After that we compute $hub(\psi)=(10,3,2,1,9,8,7,6,5),$ and $$
 \psi_2=anchor(\psi)=son(hub(\psi),3)=(10,2,1,9,8,7,4,6,5),$$
 $$height(\psi_2)=\Delta(n-3,3)=5,\
rank(\widetilde{\psi_2})-rank(\widetilde{\psi_1})=SUM(5,5)=83246,$$
$$ height(\psi_1)=\Delta(5,5)=3,rank(\widetilde{\psi_1})-rank(\widetilde{\psi})=SUM(3,4)=1955.$$

\begin{figure}[h!]
\centerline{\includegraphics[width=10.5cm]{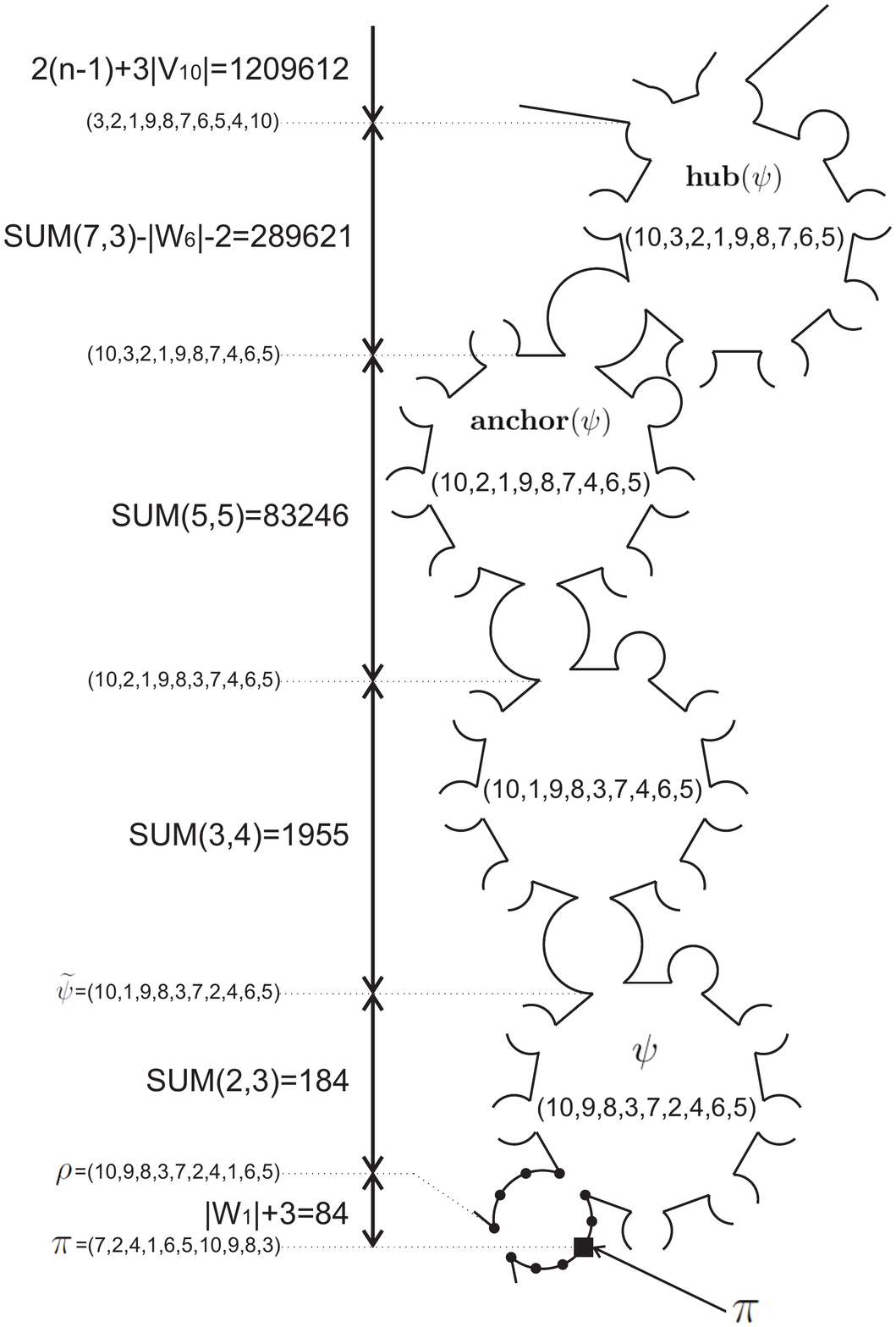}}
%\vspace*{0.2cm}
\caption{Illustration of ranking and unranking of $\pi=(7,2,4,1,6,5,10,9,8,3)$.
We have $\pi\in perms(\psi)$, and $route(\psi)=(3,5,4)$.
}\label{fig: rank_example}
\end{figure}

Now it is enough to compute $rank(\widetilde{\psi_2})$ (knowing that $\widetilde{\psi_2}$ belongs to the hub), which again is computed in two steps -- 
rank of the first permutation of $perms(hub(\psi))$ (outside of the starting path)
is equal to $2n-2+3|V_{10}|=1209612$ and $\widetilde{\psi_2}$ occurs $SUM(7,3)-|W_6|-2=289621$ permutations later. After summing all the values  our final output is:

\smallskip
\centerline{$rank(7,2,4,1,6,5,10,9,8,3)\;=\;1584702.$}
\paragraph{\bf Unranking.} Forget now that we already know the permutation with rank 1584702.
When looking for a permutation with rank $t$ we first check if $Perm(t)$ is not in the 
starting path
($t>2n-2$) and then after subtracting $2n-2=18$ from $t$ we divide it by $|V_{10}|$,
to get $t_1=3,t_2=375090$. We now know, that the permutation belongs to
$bunch((10,3,2,1,9,8,7,6,5))$. 

We have $$SUM(7,3)-|W_6|-5\le375090<SUM(7,4)-|W_6|-6,$$ hence we know,
that the permutation belongs to the $\sigma_3W_{\Delta(7,3)}\gamma_5$ part of $V_{10}$. We 
decrease the rank by $SUM(7,3)-|W_6|-2=289621$ to get $85469$.

\smallskip
Then we descend down the seed tree by choosing the fifth son, because 
 $SUM(5,5)-5\le85469<SUM(5,6)-6$, with the remaining rank $85469-SUM(5,5)=2223$.
 Next we go to 
the third son since $SUM(3,4)-4\le2223<SUM(3,5)-5$, with the remaining rank equal  $2223-SUM(3,4)=268$.

\smallskip
In the next step we know that $SUM(2,3)-3\le 268<SUM(2,4)-4$, and 
also that $268>SUM(2,3)+|W_1|$, hence further descent is not needed.

In this moment we came to an unknown seed $\psi$ for which we know $route(\psi)=(3,5,4)$, and $anchor(\psi)=son((10,3,2,1,9,8,7,6,5),3)=(10,2,1,9,8,7,4,6,5)$.
Using Claim \ref{Nov15} we recover $\psi=(10,9,8,3,7,2,4,6,5)$.
Now we know that the required permutation is in $perms(\psi)$, and
it equals $Perm(\psi,268)$, then we use Lemma \ref{1case - unranking} to obtain
$Perm(\psi,268)=(7,2,4,1,6,5,10,9,8,3)$, and this permutation is our final output.
\newpage
\section{Cyclic \texorpdfstring{$\sigma\tau$}{sigma-tau}-sequence}
%A $\sigma\tau$-sequence of permutations is cyclic if the first permutation 
%can be obtained from the last one by a single $\sigma$-move or $\tau$-move.
A $\sigma\tau$-sequence of permutations is cyclic if the last permutation is equal to
the first one.
Sawada and Williams in \cite{sawadahomepage} have given an iterative construction 
of a cyclic $\sigma\tau$-sequence. They have shown how to partition the graph
of permutations into two edge disjoint cycles (2-cycles cover)
$C',C''$ of respectively {\it inner} and {\it outer} cycles. Below we give an example 
of this structure for $n=7$.

\medskip
\centerline{\includegraphics[width=12cm]{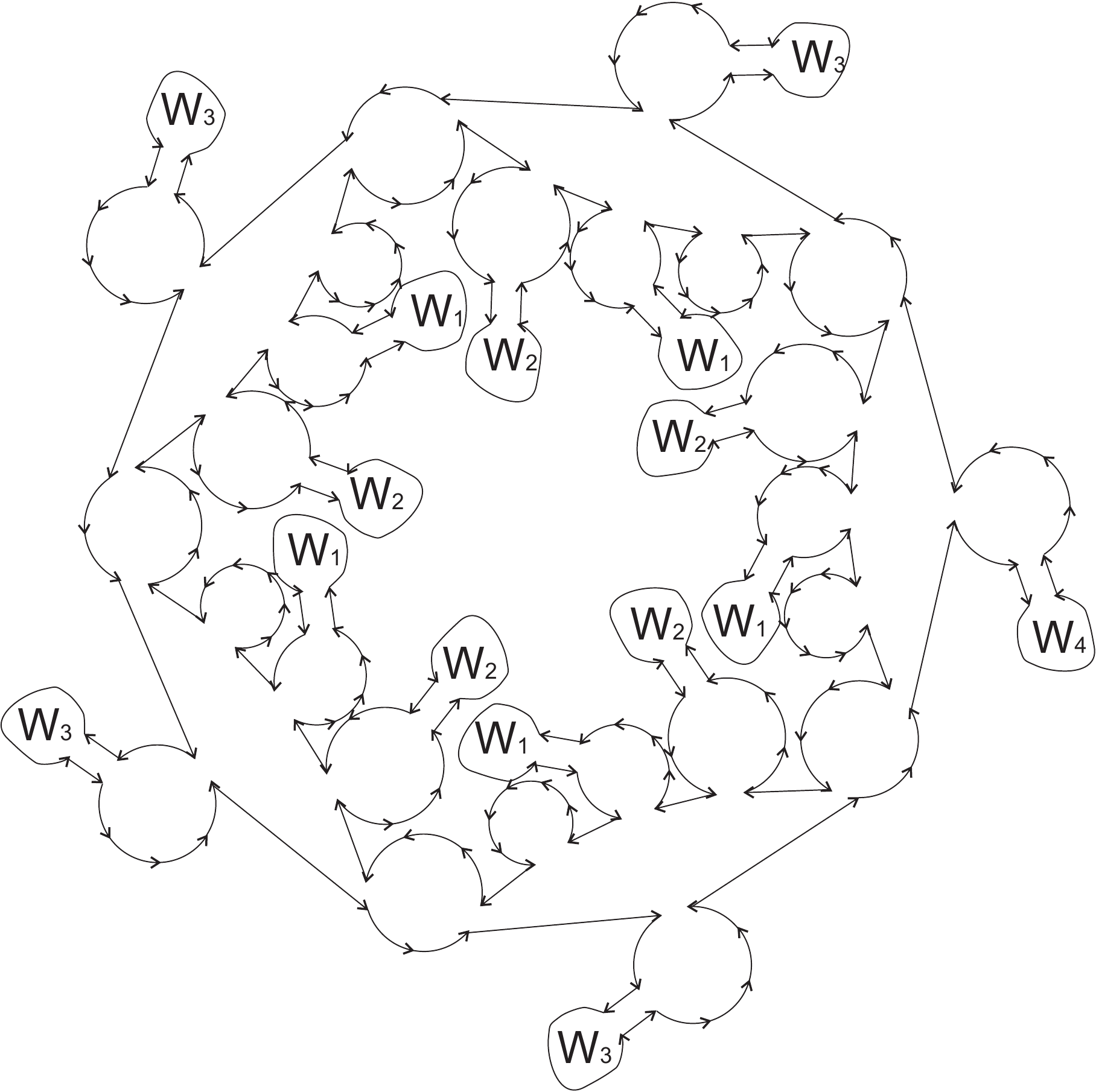}}

\medskip
\noindent We define "switches" as permutations of the form
$$(x,n,x\oplus 1,x\oplus2,...,x\ominus 1)$$
for $1\le x < n$. In other words they are cyclic shifts of the identity 
permutation $(1,2,3,..., n-1)$ in which $n$ is inserted 
into the second position.

\begin{observation}
There is a one switch on the inner cycle $C'$ and $n-2$ switches on the 
outer cycle. $C''$
\end{observation}

\subsection{\bf Sawada Williams construction of Hamiltonian cycle.}
The algorithm basically computes the cycles $C', C''$, then 
the switches are appropriately ordered 
as $\Delta_1,\, \Delta_2,\, ...,\Delta_{n-1}$.
Afterwards the outgoing edges for the switches are redirected by choosing the outgoing $\tau$-edge (in the 2-cycle cover these were $\sigma$-edges).
More explicitly we redirect:

for  each $i>1$ redirect
%$\Delta_i\Arrow{\tau} \sigma(\Delta_{i+1})$, 
$\Delta_i\overrightarrow{\tau} \sigma(\Delta_{i-1})$, %TODO \Arrow się u mnie nie kompilowało
 additionally 
redirect %$\Delta_{n-1}\Arrow{\tau} \sigma(\Delta_{1})$.
$\Delta_{1}\overrightarrow{\tau} \sigma(\Delta_{n-1})$.\newline

\subsection{Two cycles construction}

Let $\otimes$ denote a modified addition modulo $n-2$, where $n-1\otimes 1 = 2$.
It gives a cyclic order of elements $\{2,...,n-1\}$, with attached element $1$ ($1\otimes1=2$). For $a\neq 2$ we write $a\oslash 1\,=a-1$ and $2\oslash 1\,=n-1$.
\medskip
\begin{lemma}\label{two cycles representation}
The outer sequence $C''$ is represented (after removing one edge) 
by the linear sequence
$$SEQ''\;=\; (\sigma W_{n-3}\gamma_{n-3}\gamma_2) (\sigma 
W_{n-4}\gamma_{n-3}\gamma_2)^{n-3},$$
 and 
the inner sequence $C'$ is represented by the linear
sequence $SEQ'=U^{n-2}$,
 where $$U =\; \gamma_{n-4}\ \cdot\ \prod_{i=3}^{n-3}\,\sigma^i\, \W_{\Delta(n-3,i)}\,\gamma_{n-2-i} \ \cdot \ \gamma_{n-1};$$\newline
 %\TODO{wybralem dwie nazwy dla definicji, zeby pozniej ktoras wybrac. U jest troche prostsze i dobrze wkomponowuje sie w koncepcje "nowego V",
 %zas V' podkresla to, ze jest to czesc normalego $V_n$ (po usunieciu pierwszego syna)}
 %\TODO{oba wzory na cykle zakladaja brak usuniecia krawedzi (przejscie po sciezkach daje powrot do miejsca startu)
 %- pierwszy wzor jest latwo zmienic tak zeby usunac jedna z krawedzi, ale drugi juz nie wyjdzie $U^{n-2}$ tylko jakos brzydziej
 %(trzeba by wyciagnac cos poza U i usunac z jednego powtorzenia)}
\end{lemma}
\begin{proof}
 First we need to prove, that for $\psi$ outside of hub $$GEN(\widetilde{\psi},\, \W_k)\,=\, bunch(\psi)\ \ \mbox{and}\ \ 
 \W_k(\widetilde{\psi})=\psi^{(n-1)},$$ for $bunch(\psi),\widetilde{\psi},\psi^{(n-1)}$ and $height(\psi)$ defined as before, but after replacing $\oplus$ with $\otimes$.
 For $k\ge1$ the ordering of elements (which element is considered missing in the seed, and thus equal to $p_2\oplus1$ or $p_2\otimes1$) is in fact inherited from
 the parent seed (the first $k$ elements after $n$ in the permutation), thus there is no change from the proof of \ref{recurrences for W_k}.
 For $"k=0"$ (the cycle of permutations which belong to $perms(\psi)$ for just one $\psi$) the missing element is inserted just after $n$,
 thus there can be no further descent in the tree of seeds and at the same time there is only one permutation in the cycle for which
 the SW function gives a $\tau$-edge.\newline
 %\TODO{ogolnie chodzi o to, ze porzadkiem na elementach dla $W_k$ nie jest tak na prawde $\otimes$ czy $\oplus$,
 %tylko porzadek nadany z ojca a pierwszy raz gdy to sie nie pokrywa nastepuje gdy osiagamy poziom na ktorym nie schodzimy glebiej - w tym miejscu
 %zmiana reguly daje zmiane drzewa - w glownej czesci pracy to bylo glownie napisane jako "to widac" i $\widetilde{\beta_i}=\sigma^i(\psi^{(i)})$
 %-- tutaj napisalem to dokladniej, ale pewnie skrocimy do czegos takiego jak tam -- dla nowych definicji dowody sie nie roznia nawet dla missing = 2, 3 czy $n-1$}
 Every hub seed has one child which is also a hub seed. In the previous construction it was always $son(\psi,1)$. In this construction it is $son(\psi,2)$ as hub seeds
 are those of a form $(n,x,x\oslash1,...,x\otimes2,1)$. Hence the construction divides the first son from sons $3$ to $n-3$ (and the cycle with $W_0$).
 The {\it outer} cycle covers the first sons of hub seeds, and the {\it inner} cycle covers the remaining ones.
 Son of each hub seed have height $n-4$ with one exception -- seed $(n,n-3,n-4,...,1,n-1)$ which is the first son of a hub seed $(n,n-2,n-3,...,1)$ has height $n-3$.
 %\TODO{trzeba na nowo zdefiniowac wysokosc, bo ta definicja sie zmienila}
 Hence the outer cycle has the stated representation (additional $\gamma_2$ represents transition to the second son -- next hub seed).
 In the inner cycle each child of hub seeds has the same height as in the previous construction ($height(son(\psi,i))=\Delta(n-3,i)$)
 and are visited in the same order. After all those children are visited
 there appears $\gamma_{n-4}$ which represents the "return to the parent seed" (thus in this cycle hub seeds are visited in the reversed order).
\end{proof}

\subsection{Alternative path construction}
\begin{claim}
 In the path obtained with the new method the sequence representing it is equal to
 $$(\sigma W_{n-3}\gamma_{n-3}\gamma_2) (\sigma W_{n-4}\gamma_{n-3}\gamma_2)^{n-4}\sigma W_{n-4}\gamma_{n-3}\sigma^{3}U^{n-2}$$ (with the last $\tau$ removed)
 It is equal to the concatenation of representations from Lemma \ref{two cycles representation} with the ending $\tau$ removed in both representations and with $\sigma$
 added between them.
\end{claim}
\begin{lemma}
 In the new path we can rank a permutation in $\Oh(n\sqrt{\log n})$ time and unrank it in $\Oh(n\frac{\log n}{\log\log n})$.
\end{lemma}
\begin{proof}
 Ranking algorithm in the previous construction contained four parts:
 \begin{enumerate}
  \item Counting $rank(\pi)-rank(\psi)$ for $\pi\in perms(\psi)$ in $\Oh(n)$.
  \item Computing $route(\psi)$ for a given $\psi$ in $\Oh(n\sqrt{\log n})$.
  \item Counting $rank(\psi)-rank(anchor(\psi))$ out of $route(\psi)$ in $\Oh(n)$.
  \item Ranking inside the hub in $\Oh(n)$.
 \end{enumerate}
 With a few minor changes we can adjust it to the new path construction.\newline
 The first part does not really change (it is enough to know what are the missing elements in the seed and its children).\newline
 In the second part the only difference is that in the order used in inversion vector problem we must insert element $1$ somewhere. As it is never a missing element, we can
 insert it as a minimal element (or leave the order untouched if $a_2=1$).\newline
 The third part is identical.\newline
 The biggest difference occurs in the fourth part -- we first need to determine to which cycle from the two cycle construction it belongs to.
 The permutation belongs to the outer cycle when it belongs to perms of two seeds $\psi,\phi$, where $\phi=parent(\psi)$,
 and either $\psi=son(\phi,1)$, or ($\psi=son(\psi,2)$ and it is one of three first permutations in the cycle).
 In this case we must rank the permutation in relation to $\widetilde{\phi}$
 (rank in $perms(\phi)$ + $|W_{n-4}|$ (or $|W_{n-3}|$) unless it is $\widetilde{\phi}$ or $\sigma(\widetilde{\phi})$) and add
 $$rank(\widetilde{\phi})=|W_{n-3}|+n+3+(|W_{n-4}|+n-2)\cdot(n-x-2)$$
 (where $x=mis(\phi)$) if $x\neq n-1$, and $1$ if $x=n-1$.\newline
 Otherwise it belongs to the inner cycle. In this case we rank it like in the normal construction (in relation to the first permutation in $perms(\phi)$), then we subtract
 $|W_{n-4}|+n+3$ (or $|W_{n-3}|+n+3$) and add $$(x-2)|U|+|W_{n-3}|+n+3+(n-3)\cdot(|W_{n-4}|+n+2)$$.
 %\TODO{to trzeba bedzie albo jakos ladniej wyjasnic, albo jakos zamiesc pod dywan - tutaj pokazuje jak to dokladnie bedzie,
 %ale jak ktos przeczytal reszte pracy, to najprawdopodobniej sam zgadnie, ze to bedzie jakas taka wartosc
 %(no chyba ze ktos chcialby to rzeczywiscie implementowac i potrzebne mu sa dokladne stale)}
 
 Unranking algorithm contained four parts as well:
 \begin{enumerate}
  \item Finding appropriate tree of seeds (or returning permutation if it belongs to hub) in $\Oh(n)$.
  \item Computing $route(\psi)$ and $rank(\psi)$ in $\Oh(n\log\log n)$.
  \item Obtaining $\psi$ out of $route(\psi)$ in $\Oh(n\frac{\log n}{\log\log n})$.
  \item Unranking in $perms(\psi)$ in $\Oh(n)$.
 \end{enumerate}
 Parts 2,3,4 remain unchanged (the only change is the use of $\otimes$ instead of $\oplus$).\newline
 In the first part we first $$\text{determine whether\quad} t<|W_{n-3}|+n+3+(n-3)\cdot(|W_{n-4}|+n+2)$$ if that is the case we check if $t<|W_{n-3}|+n+3$
 (we unrank in the part with $W_{n-3}$) and if that is not the case we divide $t-|W_{n-3}|-n-2$ by 
 $|W_{n-4}|+n+2$ (using integer division) and unrank it in appropriate three of seeds.
 $$\text{If\quad} t'=t-(|W_{n-3}|+n+3+(n-3)\cdot(|W_{n-4}|+n+2))\ge 0$$ we unrank $t'$ in the inner cycle --
 we divide $t'$ by $|U|$ (using integer division) and proceed as in the previous algorithm.
 %\TODO{tutaj tez troche powierzchownie, ale to jest jeszcze latwiejsze do zauwazenia i ten dowod nie wymaga duzo do zalatania}
\end{proof}

\subsection{Polynomial construction for the cycle}
 \medskip
\centerline{\includegraphics[width=12cm]{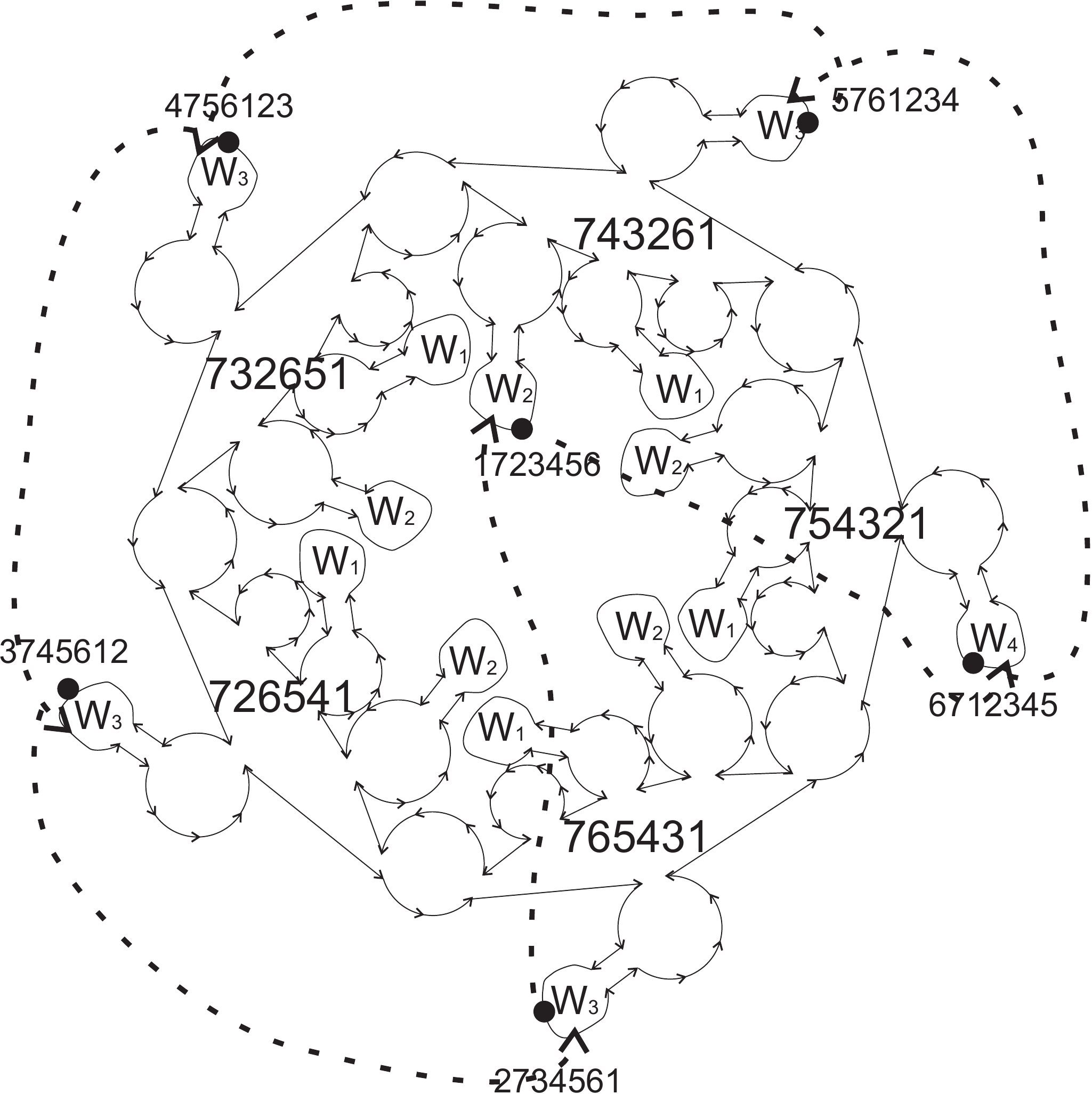}}

\medskip
\begin{lemma}
 There exist an SLP for Hamiltonian cycle of size $\Oh(n^3)$.
\end{lemma}
\begin{proof}
 Switch $(x,n,x\oplus1,...,x\ominus 1)\in perms(\psi)$ for $\psi=(n,x\oplus1,x\oplus3,...,x\ominus1,x)$. For $x\neq1$ $hub(\psi)=(n,x\oslash1,x\oslash2,...,x\otimes1,1)$,
 and $anchor(\psi)=son(hub(\psi),1)$, thus each $W_{n-4}$ (or $W_{n-3}$) on the outer cycle contains one such switch.
 For $\psi=(n,2,4,...,n-1,1)$ $hub(\psi)=(n,n-3,n-4,...,2,n-1,1)$ and $anchor(\psi)=son(hub(\psi),3)$,
 thus the remaining switch belongs to one of the $U$ parts of the inner cycle.\newline
 We can divide each such part into two -- the one before the switch and the one after it. Each such part can be represented by an $SLP$ of size $\Oh(n^2)$
 as a word $W_k$ can be divided at most once for each $k$ and $x$ (each other does not contain a switch, thus remain undivided).
\end{proof}
\begin{lemma}\label{rank in worse time}
 We can rank and unrank in the cycle in $\Oh(n^2\cdot\sqrt{\log n})$.
\end{lemma}
\begin{proof}
 Scheme of algorithm:\newline
 We count ranks for all the "switches" in the path in $\Oh(n^2\cdot \sqrt{\log n})$. Then we count differences between ranks of two next "switches" (in the order of the path)
 and ranks of switches in the cycle (iterating through switches in the order of the cycle) in $\Oh(n)$ total time and space.\newline
 Rank:
 \begin{enumerate}
  \item Rank permutation in path in $\Oh(n\cdot\sqrt{\log n})$.
  \item Find last "switch" with smaller or equal path rank, and count the difference (if rank is smaller then the rank of first switch we count everything modulo $n!$) in $\Oh(n)$.
  \item Add the difference to cycle rank of that "switch".
 \end{enumerate}
 Unrank:
 \begin{enumerate}
  \item Find last "switch" with smaller or equal cycle rank and count the difference in $\Oh(n)$.
  \item Add the difference to path rank of that "switch".
  \item Unrank in the path with the new value in $\Oh(n\frac{\log n}{\log \log n})$.
 \end{enumerate}
\end{proof}

\subsection{Efficient construction for cycle}
\begin{lemma}
 Routes of (seeds which perms contain) "switches" are always of the form $(n-3,n-4,...,1)$ with $0$ or $1$ element erased (for example $(6,5,3,2,1)$ for $n=9$)
 or are equal to $(n-3,n-4,...,3)$.
 Furthermore we can compute all the ranks of "switches" on the new path in $\Oh(n)$ total time.
\end{lemma}
\begin{proof}
 Each "switch" $\pi=(x,n,x\oplus1,...,x\ominus1)$ belongs to perms of just one seed namely $\pi\in perms(\psi)$ for
 $\psi=(n,x\oplus1,x\oplus3,x\oplus4,...,x)$ ($height(\psi)=1$).\vskip 0.4cm
 When going through parent edges until reaching $hub(\psi)$ each time the missing symbol is inserted after $n$ shifting all elements till $x\oplus1$ by one, and
 erasing the element just after it (jumping over element $1$). Hence each time parent edge is used we erase element closer to right by one with the exception of the
 one time when $1$ appears just after $x\oplus1$ for the first time. In that case the element erased next is closer to the right by two.
 Thus each route is built of numbers decreasing by ones (sometimes with one number missing). Each time the height of the tree rises by exactly one.\vskip 0.4cm
 When $x\neq 1$, then $anchor(\psi)=son(hub(\psi),1)$, thus the route ends with $1$.\newline
 If $x=n-1$ then $n-1\oplus1=1$, thus $1$ is never jumped over. As we start in a seed with tree of height $1$, 
 $route(\psi)=(n-3,n-4,...,1)$ (this is the switch which appears in the $W_{n-3}$ part).\newline
 If $1<x<n-2$, then $1$ is jumped over at the $x$-th step from $hub(\psi)$ resulting in $route(\psi)=(n-3,...,x+1,x-1,...,1)$.\newline
 When $x=n-2$ the jump over appears inside $\psi$, and thus does not touch $route(\psi)=(n-4,...,1)$.\newline
 %\TODO{ ten przypadek mozna by bylo polaczyc z poprzednim gdybysmy rozszerzyli nieco definicje route o mozliwosc dodania na poczatku numeru cyklu
 %(czyli ktorym synem jest $W_0$) - nie jest to mozliwe bez rozszerzenia bo cykl (czyli to $W_0$) nie jest seedem,
 %wtedy tez te routy staja sie unikalne (kazdy wystepuje dokladnie raz). W konstrukcji SLP, nie da to jednak zmniejszenia ilosci przypadkow, bo
 %numer w cyklu juz spada o 1 zamiast rosnac}
 When $x=1$ $route(\psi)=(n-3,...,3)$ (this is the switch from $W_{n-5}$ in the inner cycle).\newline
 
 When $x\neq n-2$\quad $\pi$ appears on the $n-3$-rd place of $n-2$-nd (last) cycle in $perms(\psi)$.
 When $x=n-2$ the jump over appears inside $\psi$ and thus $\pi$ appears  on the $n-4$-th place of $n-3$-rd cycle in $perms(\psi)$.
 
 As the heights of trees grow always by $1$ $$rank(\psi)-rank(anchor(\psi))=\sum_{i=0}^{|route(\psi)|-2} SUM(i+2,route(\psi)[i]).$$
 For $1<x<n-3$ these two values for $x$ and $x+1$ differs by $$SUM(n-x-2,x+1)-SUM(n-x-2,x)=|W_{n-x-3}|+n.$$
 Hence all the $rank(\psi)-rank(anchor(\psi))$ can be counted in $\Oh(n)$ total time
 (we get 4 cases - each counted in $\Oh(n)$ (see Lemma \ref{anchor-rank}), plus $\Oh(n)$ time
 to compute the other ones by adding the differences).
 Each $anchor(\psi)$ is the first in a part of the sequence ($U$ or $(\sigma W\gamma_{n-3}\gamma_2)$), thus all $rank(anchor(\psi))$ can be counted in $\Oh(n)$ total time.
 $rank(\pi)-rank(\psi)=(n-3)\cdot(n+1)$ ( or $(n-4)\cdot(n+1)$ for $x=n-2$).
\end{proof}
\begin{theorem}
 $\sigma\tau$-cycle from \cite{sawadahomepage} has a $\Oh(n^2)$ SLP representation. 
\end{theorem}
\begin{proof}
 Each tree of seeds in the outer cycle (and one in the inner) is divided into two parts. Each tree is divided differently, but in each $W_k$ the division happens
 either in the last $W_{k-1}$ part, or the penultimate one, thus it is enough to divide the definitions of $W_k$ into three parts
 $W_k=W_k^{(1)}(W_{k-1}\gamma_k\sigma^{n-k-1}W_{k-1})W_k^{(2)}$, and for each tree divide one of the $W_{k-1}$ words.\newline
 
 More precisely for ($2\le k\le n-3$) let:\newline
 $$W_k^{(1)}=(\tau\prod_{i=1}^{n-k-4}\sigma^iW_{\Delta(k,i)}\gamma_{n-2-i})\sigma^{n-k-2}$$
 $$W_k^{(2)}=\gamma_{k-1}\sigma^{n-k}W_{k-2}W_{k-1}^{(2)}\quad (W_1^{(2)}=\tau)$$
 $$W_k^{(3)}=W_k^{(1)}W_{k-1}\gamma_k\sigma^{n-k-1}\quad\quad
 W_k^{(4)}=\gamma_k\sigma^{n-k-1}W_{k-1}W_k^{(2)}\quad( =W_{k+1}^{(2)})$$
 The division of the tree with $(x,n,x\oplus 1,...,x\ominus 1)$ for $2\le x\le n-3$:\newline
 $$V_n^{(x,1)}= W_{n-4}^{(1)}W_{n-5}^{(1)}...W_{n-x-1}^{(1)}W_{n-x-2}^{(3)}...W_2^{(3)}(\gamma_{n-1})^{n-3}\sigma^{n-3}$$
 $$V_n^{(x,2)}=\gamma_1W_2^{(2)}...W_{n-x-2}^{(2)}W_{n-x-1}^{(4)}...W_{n-4}^{(4)}$$
 $$V_n^{(1,1)}=W_{n-5}^{(3)}...W_2^{(3)}(\gamma_{n-1})^{n-3}\sigma^{n-3}\quad\quad
 V_n^{(1,2)}=\gamma_1W_2^{(2)}...W_{n-4}^{(2)}$$
 $$V_n^{(n-2,1)}=W_{n-4}^{(1)}...W_2^{(1)}(\gamma_{n-1})^{n-4}\sigma^{n-4}\quad\quad
 V_n^{(n-2,2)}=\gamma_2\gamma_{n-1}W_2^{(4)}...W_{n-4}^{(4)}$$
 $$V_n^{(n-1,1)}=W_{n-3}^{(3)}...W_2^{(3)}(\gamma_{n-1})^{n-3}\sigma^{n-3}\quad\quad
 V_n^{(n-1,2)}=\gamma_1W_2^{(2)}...W_{n-3}^{(2)}$$
 $$C=V_n^{(1,2)}U^{n-3}\gamma_{n-4}\sigma_3V_n^{(1,1)}\tau
 \prod_{i=0}^{n-3} (V_n^{((n-1)\oslash 2i,2)}\gamma_{n-3}\gamma_2\sigma^1V_n^{((n-2)\oslash 2i,1)}\tau)$$
 %\TODO{poprawnosc wynika z dowodu poprzedniego lematu - tutaj tylko zapisuje to co tam sie wydarzylo w postaci SLP, w $W_k^{(2)}$ można
 %nie odwolywac sie rekurencyjnie do $W_{k-1}^{(2)}$ tylko rozpisac jako produkt, ale jak ktos przeczytal reszte pracy, to wie, ze ogony drzew sa takie same
 %(uzywane mocno w wyszukiwaniu w $\Oh(n\log\log n)$) i wzor jest prostszy}
\end{proof}

\begin{theorem}
 In the cycle we can rank in $\Oh(n\sqrt{\log n})$ time and unrank in $\Oh(n\frac{\log n}{\log \log n})$ time.
\end{theorem}
\iffalse
\begin{proof}
 Scheme:
 We count ranks for all the "switches" in $\Oh(n)$. Then we count differences between ranks of two next "switches" (in the order of the path)
 and ranks of switches in the cycle (iterating through switches in the order of the cycle) in $\Oh(n)$ total time and space.\newline
 Rank:
 \begin{enumerate}
  \item Rank permutation in path in $\Oh(rank)$.
  \item Find last "switch" with smaller or equal path rank, and count the difference (if rank is smaller then rank of first switch we count everything modulo $n!$) in $\Oh(n)$.
  \item Add the difference to cycle rank of that "switch".
 \end{enumerate}
 Unrank:
 \begin{enumerate}
  \item Find last "switch" with smaller or equal cycle rank and count the difference in $\Oh(n)$.
  \item Add the difference to path rank of that "switch".
  \item Unrank in the path with the new value in $\Oh(unrank)$.
 \end{enumerate}
\end{proof}
\fi
\begin{proof}
 We use the algorithm from \ref{rank in worse time}, just count ranks of "switches" in $\Oh(n)$ total time.
\end{proof}

\bibliographystyle{plainurl}
\bibliography{sigtau}

\begin{thebibliography}{1}

\bibitem{DBLP:conf/soda/ChanP10}
Timothy~M. Chan and Mihai Patrascu.
\newblock Counting inversions, offline orthogonal range counting, and related
  problems.
\newblock In Moses Charikar, editor, {\em Proceedings of the Twenty-First
  Annual {ACM-SIAM} Symposium on Discrete Algorithms, {SODA} 2010, Austin,
  Texas, USA, January 17-19, 2010}, pages 161--173. {SIAM}, 2010.

\bibitem{DBLP:conf/wads/Dietz89}
Paul~F. Dietz.
\newblock Optimal algorithms for list indexing and subset rank.
\newblock In {\em Algorithms and Data Structures, Workshop {WADS} '89, Ottawa,
  Canada, August 17-19, 1989, Proceedings}, pages 39--46, 1989.

\bibitem{DBLP:journals/talg/RuskeyW10}
Frank Ruskey and Aaron Williams.
\newblock An explicit universal cycle for the (\emph{n}-1)-permutations of an
  \emph{n}-set.
\newblock {\em {ACM} Trans. Algorithms}, 6(3):45:1--45:12, 2010.

\bibitem{DBLP:conf/icalp/Rytter04}
Wojciech Rytter.
\newblock Grammar compression, {LZ}-encodings, and string algorithms with
  implicit input.
\newblock In {\em Automata, Languages and Programming: 31st International
  Colloquium, {ICALP} 2004, Turku, Finland, July 12-16, 2004. Proceedings},
  pages 15--27, 2004.

\bibitem{sawadahomepage}
Joe Sawada and Aaron Williams.
\newblock Solving the sigma-tau problem.
\newblock URL: \url{http://socs.uoguelph.ca/~sawada/papers/sigmaTauCycle.pdf}.

\bibitem{DBLP:conf/soda/SawadaW18}
Joe Sawada and Aaron Williams.
\newblock A {Hamilton} path for the sigma-tau problem.
\newblock In {\em Proceedings of the Twenty-Ninth Annual {ACM-SIAM} Symposium
  on Discrete Algorithms, {SODA} 2018, New Orleans, LA, USA, January 7-10,
  2018}, pages 568--575, 2018.

\bibitem{DBLP:journals/dm/SawadaWW16}
Joe Sawada, Aaron Williams, and Dennis Wong.
\newblock A surprisingly simple {de Bruijn} sequence construction.
\newblock {\em Discrete Mathematics}, 339(1):127--131, 2016.

\end{thebibliography}
\end{document}